\documentclass[journal,final,twocolumn,twoside]{IEEEtran}
\usepackage{amsmath,amsfonts}
\usepackage{algorithmic}
\usepackage{algorithm}
\usepackage{amssymb}
\usepackage{graphicx}
\allowdisplaybreaks[4]
\usepackage{float} 
\usepackage{multirow} 
\usepackage{ntheorem}
\usepackage{balance}
\usepackage{enumitem}
\usepackage{caption}
\usepackage{arydshln}
\usepackage{upgreek}
\usepackage{setspace}
\usepackage{epsf}
\usepackage{upgreek}
\usepackage{balance}
\usepackage{footnote}
\usepackage{color}
\usepackage{setspace}
\usepackage{array}

\usepackage{textcomp}
\usepackage{stfloats}
\usepackage{url}
\usepackage{bm}
\usepackage{verbatim}
\usepackage{ntheorem}
\usepackage{cite}
\usepackage{subcaption}

\theorembodyfont{\upshape}
\theoremheaderfont{\it}
\qedsymbol{\square}
\theoremseparator{:}

\newtheorem*{proof}{\indent Proof}

\newtheorem{proposition}{\indent Proposition}

\makeatletter

\newcommand{\Rmnum}[1]{\expandafter\@slowromancap\romannumeral #1@}
\makeatother

\begin{document}

\title{Flexible-Duplex Cell-Free Architecture for Secure Uplink Communications in Low-Altitude Wireless~Networks}

\author{Wei~Shi, Wei~Xu,~\IEEEmembership{Fellow,~IEEE}, Yongming~Huang,~\IEEEmembership{Fellow,~IEEE}, Jiacheng~Yao, Wenhao Hu, \\and Dongming~Wang,~\IEEEmembership{Member,~IEEE} 

\thanks{W. Shi, W. Xu, Y. Huang, J. Yao, W. Hu, and D. Wang are with the Purple Mountain Laboratories, Nanjing 211111, China (e-mail: shiwei1@pmlabs.com.cn). W. Xu, Y. Huang, J. Yao, W. Hu, and D. Wang are also with the National Mobile Communications Research Laboratory, Southeast University, Nanjing 210096, China (e-mail: \{wxu, huangym, jcyao, wenhaohu, wangdm\}@seu.edu.cn).}

}

\maketitle

\begin{abstract}

Low-altitude wireless networks (LAWNs) are expected to play a central role in future 6G infrastructures, yet uplink transmissions of uncrewed aerial vehicles (UAVs) remain vulnerable to eavesdropping due to their limited transmit power, constrained antenna resources, and highly exposed air-ground propagation conditions. To address this fundamental bottleneck, we propose a flexible-duplex cell-free (CF) architecture in which each distributed access point (AP) can dynamically operate either as a receive AP for UAV uplink collection or as a transmit AP that generates cooperative artificial noise (AN) for secrecy enhancement. Such AP-level duplex flexibility introduces an additional spatial degree of freedom that enables distributed and adaptive protection against wiretapping in LAWNs. Building upon this architecture, we formulate a max-min secrecy-rate problem that jointly optimizes AP mode selection, receive combining, and AN covariance design. This tightly coupled and nonconvex optimization is tackled by first deriving the optimal receive combiners in closed form, followed by developing a penalty dual decomposition (PDD) algorithm with guaranteed convergence to a stationary solution. To further reduce computational burden, we propose a low-complexity sequential scheme that determines AP modes via a heuristic metric and then updates the AN covariance matrices through closed-form iterations embedded in the PDD framework. Simulation results show that the proposed flexible-duplex architecture yields substantial secrecy-rate gains over CF systems with fixed AP roles. The joint optimization method attains the highest secrecy performance, while the low-complexity approach achieves over $90\%$ of the optimal performance with an order-of-magnitude lower computational complexity, offering a practical solution for secure uplink communications in LAWNs.


\end{abstract}
\section{Introduction}
Low-altitude wireless networks (LAWNs), consisting of uncrewed aerial vehicles (UAVs), aerial robots, and low-altitude platform stations, are emerging as a fundamental component of next-generation intelligent infrastructures \cite{he2025ubiquitous,wu2025toward,wshi1,yuan2025ground}. By enabling agile and wide-area aerial connectivity below 3~km altitude, LAWNs support a variety of mission-critical applications, such as last-mile logistics, emergency response, urban surveillance, smart-city sensing, and precision agriculture \cite{wxu,jin2025co}. The unique three-dimensional (3D) operating space of LAWNs offers unprecedented flexibility in deployment and service, making them an indispensable enabler for the upcoming sixth-generation (6G) networks.

Despite these advantages, LAWNs differ from conventional terrestrial networks due to the high mobility of UAVs, rapid topology variations, and their strong reliance on line-of-sight (LoS) links to ensure high-quality air-ground communications. While these features enhance spectral efficiency and connectivity, they also expose UAV links to severe security vulnerabilities \cite{sun2019physical,li2019secure}. In particular, the open LoS channels in low-altitude airspace are highly susceptible to eavesdropping attacks. Most existing LAWN implementations still adopt the conventional cellular paradigm, which relies on fixed cell partitioning and centralized coordination \cite{geraci2022will}. Such designs inevitably lead to strong inter-cell interference, limited boundary coverage, and privacy leakage, which become even more pronounced under the stringent size, weight, and power constraints of UAV platforms. Consequently, traditional cellular designs are increasingly inadequate for ensuring secure UAV communications in dynamic low-altitude environments.

To overcome these limitations, the cell-free (CF) framework has emerged as a promising paradigm for UAV-enabled LAWNs \cite{elhoushy2021cell}. By allowing a large number of distributed access points (APs) to jointly serve users without cell boundaries, the CF architecture inherently enhances spatial diversity, mitigates inter-cell interference, and improves user-centric fairness \cite{wang2022joint,yao2023robust,yang2025privacy}. Recent works have investigated CF-based UAV networks from multiple perspectives. In \cite{xu2023air}, an air-ground cooperative CF architecture was proposed to jointly optimize downlink beamforming, fronthaul compression, and UAV positioning, improving both coverage and energy efficiency. The authors of~\cite{d2020analysis} developed a user-centric association mechanism under Rician fading to characterize the spectral efficiency, demonstrating the fairness and capacity benefits of the CF operation. Furthermore, \cite{zheng2021uav,tentu2022uav} investigated the impact of hardware impairments and energy efficiency optimization. In addition to network architecture design, recent studies have focused on trajectory optimization and power control in CF-based UAV networks. The authors of \cite{liu2025energy} proposed a multi-agent reinforcement learning algorithm to balance energy consumption and network capacity, while a power control scheme for ultra-reliable and low-latency communication scenarios was developed in \cite{elwekeil2022power}. Collectively, these works confirm the effectiveness of CF architectures in supporting highly dynamic and resource-constrained UAV communications.

Beyond performance optimization, the CF paradigm also offers new opportunities for enhancing physical layer security (PLS) owing to its intrinsic capability for wide-area cooperative interference management and distributed signal coordination \cite{shi2025combating}. By exploiting cooperative transmission, joint reception, and artificial noise (AN)-aided beamforming, CF systems can strengthen legitimate links while suppressing potential eavesdropping. However, most existing research on CF-based secure transmission has focused on terrestrial communication scenarios rather than low-altitude networks. For example, \cite{hoang2018cell} investigated secure transmission under pilot spoofing attacks via optimized power allocation, while \cite{zhang2021secure} analyzed secrecy performance under hardware impairments. The integration of reconfigurable intelligent surfaces for secrecy enhancement was explored in \cite{elhoushy2021exploiting}, and a stochastic-geometry-based analysis of scalable CF networks was presented in~\cite{ma2023secrecy}, where the secrecy performance was evaluated in terms of outage-based secrecy transmission rate and ergodic secrecy rate under random AP, user, and eavesdropper deployments. Although these efforts have demonstrated the potential of CF architectures for enhancing PLS, they primarily focused on terrestrial user scenarios and overlooked the distinctive characteristics of low-altitude UAV networks.

Compared with terrestrial scenarios, UAV communications exhibit strong LoS propagation, limited onboard transmit resources, and high mobility, making them substantially more vulnerable to interception. These unique attributes undermine the effectiveness of many ground-oriented PLS schemes and motivate the exploration of new secure transmission strategies tailored to low-altitude environments. Recent studies have taken initial steps toward this direction.
For instance, pilot allocation and power control in CF networks with eavesdroppers were considered in \cite{chen2023pilot}, where secrecy-rate lower bounds were derived. While these contributions illustrate the feasibility of CF-enhanced PLS in UAV scenarios, systematic investigations tailored to the highly dynamic and LoS-dominant nature of LAWNs remain limited. In particular, the cooperative role of distributed APs in simultaneously strengthening legitimate transmissions and suppressing eavesdropping in low-altitude environments has not been thoroughly explored.

To further improve the flexibility and spectral efficiency of CF architectures, the concept of network-assisted full-duplex (NAFD) has been developed as a unified duplexing paradigm that subsumes traditional half-duplex, hybrid-duplex, and co-frequency co-time full-duplex operations \cite{wang2019performance,xia2021joint}. In CF networks empowered by NAFD, each AP can dynamically operate in uplink reception or downlink transmission mode, with its duplex direction adaptively selected according to instantaneous channel conditions. This configuration, often referred to as flexible-duplex CF, enables distributed APs to cooperatively balance data transmission and interference suppression across space and time. 
Unlike conventional full-duplex systems that suffer from severe self-interference and cross-link interference, flexible-duplex CF mitigates these effects by spatially decoupling the transmit and receive functions among APs while maintaining centralized baseband coordination at the central processing unit (CPU). As demonstrated in~\cite{wang2019performance}, flexible-duplex CF networks can achieve higher spectral efficiency than both full-duplex and half-duplex configurations with the same total antenna resources, owing to macro-diversity and distributed antenna gains. Similarly, the work in~\cite{xia2021joint} showed that dynamically scheduling transmit-APs (T-APs) and receive-APs (R-APs) allows flexible-duplex CF networks to maximize the uplink–downlink sum rate under fronthaul and signal-to-interference-plus-noise ratio (SINR) constraints. These studies validate the practicality and efficiency of AP-level flexible-duplex adaptation for distributed CF systems.

Motivated by these insights, this paper leverages the flexible-duplex CF architecture to address PLS challenges in a LAWN. In the proposed design, each distributed AP can dynamically switch between receive mode (acting as an R-AP for UAV uplink data collection) and transmit mode (acting as a T-AP to emit cooperative AN for eavesdropper suppression). This AP-level mode selection strategy not only inherits the interference-mitigation and resource-efficiency benefits of NAFD, but also introduces a new spatial degree of freedom (DoF) for improving secrecy performance in low-altitude environments. The main contributions are summarized as follows.

\begin{itemize}
\item We propose a novel flexible-duplex CF architecture tailored for secure UAV uplink communications in LAWNs. Each AP dynamically operates either as an R-AP performing cooperative reception or as a T-AP transmitting spatially structured AN for eavesdropper suppression. This AP-level duplex flexibility introduces an additional spatial DoF that enables distributed and adaptive secrecy enhancement. A complete uplink–Eve signal model is established, capturing inter-user interference, AN coupling, cross-AP interference, and the spatial interference shaping effect induced by mode selection. 
\item We formulate a max–min secrecy-rate (MMSR) optimization problem that jointly determines the AP modes, receive combiners, and AN covariance matrices under practical duplexing and power constraints. This problem is highly coupled and intrinsically nonconvex. To tackle these challenges, we first derive the optimal receive combiners in closed form, enabling a simplified but equivalent reformulation of the original problem. Building upon this, we develop a penalty dual decomposition (PDD)-based algorithm for jointly optimizing the AP modes and AN, with guaranteed convergence to a stationary point. Within the PDD framework, the AN covariance matrices and AP modes are iteratively optimized using successive convex approximation (SCA) and majorization minimization (MM) techniques, respectively.
\item To further reduce computational overhead, we propose a low-complexity algorithm that sequentially determines the AP modes and AN covariance matrices. Specifically, we introduce a heuristic metric for AP mode selection that jointly accounts for signal reception quality, eavesdropping-jamming capability, and inter-AP interference, and select the AP modes accordingly. Based on the obtained AP modes, we further employ the PDD framework to address the remaining nonconvex AN optimization problem. Moreover, within each PDD iteration, we derive closed-form optimal updates for the involved variables, thereby enabling efficient and low-complexity iterative refinement.
\item Comprehensive simulations are conducted to verify the effectiveness of the proposed flexible-duplex CF architecture. The results show that jointly optimizing AP modes and AN design yields substantial secrecy-rate gains over conventional CF schemes with fixed AP roles, underscoring the importance of adaptive AP duplexing. Moreover, the proposed methods exhibit consistent fairness improvements, providing more balanced secrecy performance across UAV users compared with baselines. These findings confirm the practical value of the proposed framework for secure uplink communications in LAWNs.
\end{itemize} 

The remainder of this paper is organized as follows. Section~\uppercase\expandafter{\romannumeral2} introduces the system model and formulates the MMSR problem. Section~\uppercase\expandafter{\romannumeral3} presents the optimal receive combiners and the proposed PDD-based joint optimization algorithm. Section~\uppercase\expandafter{\romannumeral4} develops the low-complexity sequential AP mode selection and AN design scheme. Simulation results and conclusions are given in Sections~\uppercase\expandafter{\romannumeral5} and~\uppercase\expandafter{\romannumeral6}, respectively.

\emph{Notation:} Boldface lowercase (uppercase) letters denote vectors (matrices). $\mathbb{C}$ represents the set of complex numbers. The superscripts $(\cdot)^T$, $(\cdot)^\ast$, and $(\cdot)^H$ denote the transpose, conjugate, and Hermitian transpose, respectively, while $(\cdot)^{-1}$ and $(\cdot)^{\dagger}$ represent the matrix inverse and the Moore–Penrose pseudoinverse. ${\cal CN}(\boldsymbol{\mu},\mathbf{\Sigma})$ denotes a circularly symmetric complex Gaussian distribution with mean vector $\boldsymbol{\mu}$ and covariance matrix $\mathbf{\Sigma}$. $\mathbb{E}[\cdot]$ denotes the expectation of a random variable (RV). ${\rm diag}\left\{\cdot\right\}$ indicates a diagonal matrix. $|\cdot|$, $\Vert\cdot\Vert$, and $\Vert\cdot\Vert_F$ denote the modulus, Euclidean norm, and Frobenius norm, respectively. $\mathrm{Tr}(\cdot)$ denotes the matrix trace, $\odot$ denotes the Hadamard product, and $\mathbf{A}\succeq 0$ indicates that $\mathbf{A}$ is Hermitian positive semi-definite. The operator $\mathrm{vec}(\cdot)$ denotes matrix vectorization, while $\mathrm{mat}(\cdot)$ denotes its inverse operator that reshapes a vector into a matrix of appropriate dimensions.
\section{System Model}
We consider a flexible-duplex CF architecture for secure uplink communications in a LAWN, as illustrated in Fig.~\ref{fig1}. The network comprises $M$ distributed APs, each equipped with $N_a$ antennas, collaboratively serving $K$ UAVs over a shared time-frequency resource unit. All the APs are connected to a CPU via high-capacity fronthaul links, enabling joint signal processing and coordination across the entire network. 

Under the flexible-duplex operation, each AP can dynamically select between two working modes, i.e., the receive mode and the transmit mode, depending on instantaneous system requirements and channel conditions. When operating in the receive mode, the R-AP collects uplink signals transmitted by the UAVs and forwards them to the CPU for centralized decoding. In contrast, when operating in the transmit mode, the T-AP emits AN to deliberately jam the eavesdropper, thereby enhancing the PLS of uplink communications in the LAWN. 

The following subsections present the detailed signal models for the legitimate receivers and the eavesdropper, based on which the secrecy-rate optimization problem is developed.

\begin{figure}[!t]
\centering
\includegraphics[width=3.2in]{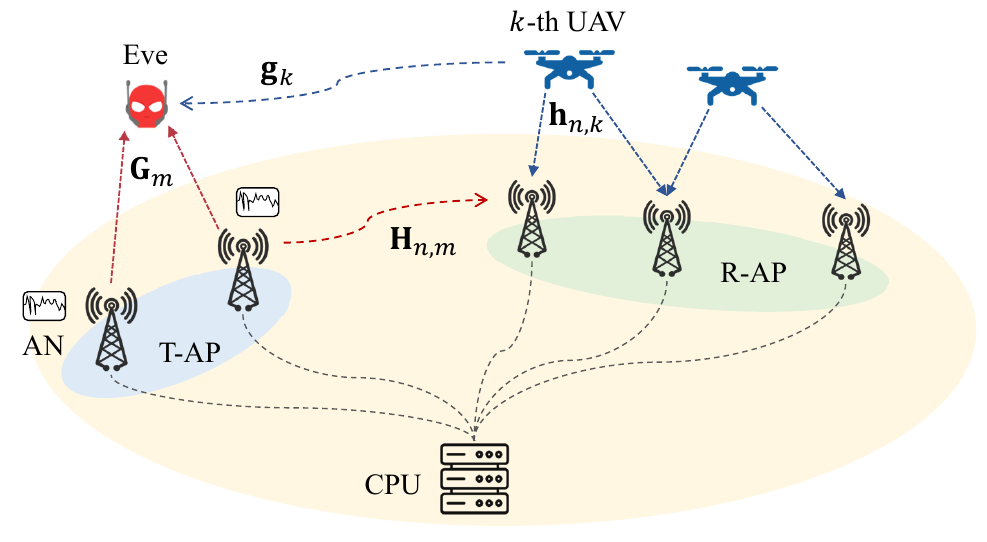}  
\caption{Illustration of a flexible-duplex CF architecture for secure uplink communications in a LAWN.}
\label{fig1}
\end{figure}

\subsection{Uplink Signal and SINR Model}
Let $x_m^C\in\left\{0,1\right\}$ and $x_m^J\in\left\{0,1\right\}$ denote binary mode-selection variables of the $m$-th AP, satisfying $x_m^C+x_m^J=1$. When $x_m^C=1$, the AP operates in the receive mode; when $x_m^J=1$, it operates in the transmit mode and transmits an AN vector $\mathbf{a}_m\sim\mathcal{CN}(\mathbf{0},\mathbf{V}_m)$. Here, $\mathbf{V}_m=\mathbb{E}\left[\mathbf{a}_m\mathbf{a}_m^H\right]\succeq\mathbf{0}$ denotes the AN covariance matrix that determines its spatial structure and power allocation.

Accordingly, the received signal at the $n$-th AP is
\begin{align}
    \mathbf{y}_n=x_n^C\left(\sum_{k=1}^{K}{\mathbf{h}_{n,k}\sqrt{p_k}s_k}+\sum_{m=1}^{M}{x_m^J\mathbf{H}_{n,m}\mathbf{a}_m}+\mathbf{n}_n\right),
\end{align}
where $p_k$ and $s_k$ denote the transmit power and data symbol of UAV $k$ with $\mathbb{E}[\left|s_k\right|^2]=1$, respectively. $\mathbf{h}_{n,k}\in\mathbb{C}^{N_a\times1}$ represents the channel vector from uplink UAV $k$ to R-AP $n$, whose large-scale fading depends on the UAV's position, $\mathbf{H}_{n,m}\in\mathbb{C}^{N_a\times N_a}$ denotes the interference channel from T-AP $m$ to R-AP $n$, and $\mathbf{n}_n\sim\mathcal{CN}(\mathbf{0},\sigma_n^2\mathbf{I}_{N_a})$ denotes the additive white Gaussian noise (AWGN) with power $\sigma_n^2$.

The APs forward all received data signals to the CPU for processing, enabling a centralized cell-free implementation. By stacking $\mathbf{y}_n$ for all $n\in \{1,2,...,M\}$, the received signal $\mathbf{y}\in\mathbb{C}^{MN_a\times1}$ at the CPU is given by
\begin{align}\label{rs_cpu}
    \mathbf{y}=\left[\mathbf{y}_1^T,\ldots,\mathbf{y}_M^T\right]^T.
\end{align}

To detect $s_k$, the CPU employs a receive combining vector $\mathbf{u}_k=\left[\mathbf{u}_{1,k}^T,\ldots,\mathbf{u}_{M,k}^T\right]^T\in\mathbb{C}^{MN_a\times1}$, and thus, the combined signal for UAV $k$ is
\begin{align}\label{rs}
    r_k=\,&\mathbf{u}_k^H\mathbf{y}=\sum_{n=1}^{M}{\mathbf{u}_{n,k}^H\mathbf{y}_n}\nonumber\\
       =\,&\sum_{n=1}^{M}\sum_{i=1}^{K}{\mathbf{u}_{n,k}^Hx_n^C\mathbf{h}_{n,i}\sqrt{p_i}s_i}\nonumber\\
       &+\sum_{n=1}^{M}\sum_{m=1}^{M}{\mathbf{u}_{n,k}^Hx_n^Cx_m^J\mathbf{H}_{n,m}\mathbf{a}_m}+\sum_{n=1}^{M}{\mathbf{u}_{n,k}^Hx_n^C\mathbf{n}_n}.
\end{align}

For notational convenience, we define a diagonal mode-selection matrix as
\begin{align}\label{MS}
\mathbf{S}_c\triangleq\mathrm{diag}\left\{x_1^C\mathbf{I}_{N_a},\ldots,x_M^C\mathbf{I}_{N_a}\right\}\in\mathbb{C}^{MN_a\times M N_a}, 
\end{align}
where each diagonal block $x_m^C\mathbf{I}_{N_a}$ indicates whether the $m$-th AP is active in the receive mode. In other words, $\mathbf{S}_c$ acts as a spatial selection operator that preserves the received signal components from APs operating in receive mode ($x_m^C=1$) while nulling those from transmit-mode APs ($x_m^C=0$).

Then, the instantaneous SINR of UAV $k$ at the CPU is derived as~(\ref{rsinr}), where 
$\mathbf{h}_k \triangleq [\mathbf{h}_{1,k}^T,\ldots,\mathbf{h}_{M,k}^T]^T \in \mathbb{C}^{MN_a\times1}$ 
denotes the stacked UAV–AP channel vector, and 
$\mathbf{H}_m \triangleq [\mathbf{H}_{1,m}^T,\ldots,\mathbf{H}_{M,m}^T]^T \in \mathbb{C}^{MN_a\times N_a}$ 
represents the stacked interference-channel matrix from the $m$-th AP to all APs.

\begin{figure*}
\begin{align}\label{rsinr}
    \gamma_k=\frac{p_k\left|\mathbf{u}_k^H\mathbf{S}_c\mathbf{h}_k\right|^2}{\sum_{i\neq k}^{K}{p_i\left|\mathbf{u}_k^H\mathbf{S}_c\mathbf{h}_i\right|^2}+\sum_{m=1}^{M}{x_m^J\mathbf{u}_k^H\mathbf{S}_c\mathbf{H}_m\mathbf{V}_m\mathbf{H}_m^H\mathbf{S}_c\mathbf{u}_k}+\sigma_n^2\mathbf{u}_k^H\mathbf{S}_c\mathbf{u}_k}.
\end{align}
\hrulefill
\end{figure*}

\subsection{Eavesdropper Signal and SINR Model}
We consider an eavesdropper (Eve) equipped with $N_e$ antennas attempting to wiretap the legitimate UAV transmissions. Then, the received signal at Eve is expressed as
\begin{align}
    \mathbf{y}_e=\sum_{k=1}^{K}{\mathbf{g}_k\sqrt{p_k}s_k}+\sum_{m=1}^{M}{x_m^J\mathbf{G}_m\mathbf{a}_m}+\mathbf{n}_e,
\end{align}
where $\mathbf{g}_k\in\mathbb{C}^{N_e\times1}$ denotes the channel from UAV $k$ to Eve, $\mathbf{G}_m\in\mathbb{C}^{N_e\times N_a}$ represents the channel from AP $m$ to Eve, and $\mathbf{n}_e\sim\mathcal{CN}(\mathbf{0},\sigma_e^2\mathbf{I}_{N_e})$ is the AWGN at Eve.

If Eve employs a linear combining vector $\mathbf{w}_{e,k}$ to detect the signal of UAV $k$, the resulting SINR is given by (\ref{esinr}).
\begin{figure*}
\begin{align}\label{esinr}
    \gamma_{e,k}=\frac{p_k\left|\mathbf{w}_{e,k}^H\mathbf{g}_k\right|^2}{\sum_{i\neq k}{p_i\left|\mathbf{w}_{e,k}^H\mathbf{g}_i\right|^2}+\sum_{m=1}^{M}{x_m^J\mathbf{w}_{e,k}^H\mathbf{G}_m\mathbf{V}_m\mathbf{G}_m^H\mathbf{w}_{e,k}}+\sigma_e^2\Vert\mathbf{w}_{e,k}\Vert^2}.
\end{align}
\hrulefill
\end{figure*}

\subsection{Achievable Rate and Secrecy Rate}
Based on the SINR expressions derived in (\ref{rsinr}) and (\ref{esinr}), the achievable rate of UAV $k$ at the CPU and the eavesdropping rate at Eve are respectively given by
\begin{align}
    R_k=\log_2\left(1+\gamma_k\right), \, R_{e,k}=\log_2\left(1+\gamma_{e,k}\right).
\end{align}

Accordingly, the instantaneous secrecy rate of UAV $k$ is expressed as \cite{shi2022secure,gao2024artificial,wshitwc}
\begin{align}
   R_k^{\mathrm{sec}}=\left[R_k-R_{e,k}\right]^+=\left[\log_2\frac{1+\gamma_k}{1+\gamma_{e,k}}\right]^+,
\end{align}
where $\left[x\right]^+=\max\{0,x\}$. This expression captures the instantaneous secrecy performance at each time slot, which is particularly relevant in low-altitude UAV communications characterized by fast time-varying channels and dynamic network topologies. Hence, optimizing the instantaneous secrecy rate allows real-time adaptation to spatial and channel fluctuations.

\subsection{Optimization Problem Formulation}
To jointly exploit the cooperative potential of distributed APs and suppress eavesdropping threats, we formulate a dynamic AP mode selection problem that maximizes the minimum secrecy rate among all UAVs. Specifically, the problem is formulated as
\begin{align}\label{problem}
   \mathop{\text{maximize}}_{\left\{x_m^C,x_m^J\right\},\left\{\mathbf{u}_k,\mathbf{V}_m\succeq\mathbf{0}\right\}}\quad&{\min_{k}{R_k^{\mathrm{sec}}}}\nonumber\\
   \text{subject to}\quad & \mathrm{C1}:x_m^C\!+\!x_m^J\!=\!1, x_m^C,\!x_m^J\!\!\in\!\!\left\{0,1\right\}, \forall m, \nonumber\\
   & \mathrm{C2}:\mathrm{Tr}\left(\mathbf{V}_m\right)\le x_m^JP_m, \forall m,
\end{align}
where the binary mode-selection variables $x_m^C$ and $x_m^J$ indicate whether AP $m$ operates in reception or jamming mode, respectively. Constraint C1 enforces the mutual exclusivity of AP modes, while C2 limits the AN power of each AP according to its maximum transmit capability $P_m$.

This MMSR formulation enforces UAV fairness by avoiding secrecy-rate outage and balancing confidentiality across dynamic LAWN links. By adaptively coordinating cooperative reception and jamming, it establishes a principled basis for the flexible-duplex CF strategy to achieve robust PLS.

\section{Proposed Joint Optimization Algorithm}\label{SEC3}
In this section, we solve problem (\ref{problem}) in an iterative manner. Specifically, we first derive the optimal receive combining vectors in closed-form expressions, which simplifies the original problem. Then, we employ the PDD method to jointly optimize the AP mode selection and AN covariance matrices.

\subsection{Closed-Form Solutions to Receive Combining}
We note that the received combining vector $\mathbf{u}_k$ only affects the uplink SINR of the $k$-th UAV, while having no impact on the other UAVs or the eavesdropping performance. Therefore, for any given mode selection and AN signal, the design criterion of $\mathbf{u}_k$ is to maximize the SINR $\gamma_k$ at the CPU. Based on this criterion, its closed-form optimal solution is provided in the following proposition.

\begin{proposition} \label{proposition1}
For any given AP mode selection $\left\{x_m^C, x_m^J\right\}_{m=1}^M$ and AN covariance matrices $\{\mathbf{V}\}_{m=1}^M$, the optimal receive combiner $\mathbf{u}_k$ for the $k$-th UAV is equal to
\begin{align}
    \mathbf{u}_k^\ast\!=\!&\left(\!\sum_{i\neq k}^{K}{p_i\mathbf{S}_c\mathbf{h}_i\mathbf{h}_i^H\mathbf{S}_c}\!+\!\!\sum_{m=1}^{M}\!{x_m^J\mathbf{S}_c\mathbf{H}_m\!\mathbf{V}_m\mathbf{H}_m^H\mathbf{S}_c}\!+\!\sigma_n^2\mathbf{I}\!\right)^{-1}\nonumber\\
    &\times\mathbf{S}_c\mathbf{h}_k.
\end{align}
\end{proposition}

\begin{proof}
    Please refer to Appendix \ref{appb}.  \hfill $\square$
\end{proof}

By substituting $\mathbf{u}_k^\ast$ into (\ref{rsinr}), the equivalent simplified SINR expression becomes
\begin{align}
    \gamma_{k}^\ast=\,&p_k\mathbf{h}_k^H\mathbf{S}_c\left(\sum_{i\neq k}^{K}{p_i\mathbf{S}_c\mathbf{h}_i\mathbf{h}_i^H\mathbf{S}_c}\right.\nonumber\\
    &\left.+\sum_{m=1}^{M}{x_m^J\mathbf{S}_c\mathbf{H}_m\mathbf{V}_m\mathbf{H}_m^H\mathbf{S}_c}+\sigma_n^2\mathbf{I}\right)^{-1}\mathbf{S}_c\mathbf{h}_k.
\end{align}
Consequently, problem (\ref{problem}) reduces to an optimization problem over AP mode selection and AN design, whose specific solution is presented as follows.
\subsection{PDD Framework for AP Mode and AN Optimization}
Firstly, to handle the complicated fractional terms, we introduce slack variables and rewrite the problem in (\ref{problem}) as
\begin{align}\label{problem13}
   \mathop{\text{maximize}}_{\substack{\left\{x_m^C,x_m^J\right\},\left\{\mathbf{V}_m\succeq\mathbf{0}\right\},\\r,\left\{\tau_k,\chi_k\right\}}}\quad&{r}\nonumber\\
   \text{subject to}\quad & \mathrm{C1}-\mathrm{C2}, \nonumber\\
   & \mathrm{C3}:1+\gamma_k\geq2^{\tau_k}, \forall k,\nonumber\\
   & \mathrm{C4}:1+\gamma_{e,k}\le2^{\chi_k}, \forall k, \nonumber\\
   & \mathrm{C5}:\tau_k-\chi_k\geq r, \forall k,
\end{align}
where $r,2^{\tau_k}$, and $2^{\chi_k}$ respectively replace $\min_k R_{k}^{\text{sec}}$, $1+\gamma_k$, and $1+\gamma_{e,k}$ in the objective function. 

In the above problem, the AN covariance matrix $\mathbf{V}_m$ and $x_{m}^C$, as well as $x_m^J$, are still highly coupled, making it difficult to handle. To this end, we defining the following variables ${\overline{\mathbf{V}}}_m\!\triangleq \!x_m^J\mathbf{V}_m$ and $\mathbf{B}_k\!\triangleq\!\sum_{i\neq k}^{K}{p_i\mathbf{S}_c\mathbf{h}_i\mathbf{h}_i^H\mathbf{S}_c}\!\!+\!\!\sum_{m=1}^{M}\!{x_m^J\!\mathbf{S}_c\mathbf{H}_m\!\!\mathbf{V}_m\mathbf{H}_m^H\mathbf{S}_c}\!\!+\!\!\sigma_n^2\mathbf{I}$. Due to the introduction of ${\overline{\mathbf{V}}}_m$, $\mathbf{V}_m$ and $x_{m}^J$ are unified into one variable, and the relevant constraints are rewritten as follows:
\begin{align}
    &\mathrm{\overline{C1}}: \delta\left(\mathrm{Tr}({\overline{\mathbf{V}}}_m)\right)+x_m^C\le1,\,x_m^C\in \{0,1\},\, \forall m, \nonumber\\ &\mathrm{\overline{C2}}:\mathrm{Tr}\left({\overline{\mathbf{V}}}_m\right)\le P_m, \forall m,
\end{align}
where $\delta(\cdot)$ outputs 1 if and only if the input is non-zero, otherwise it is $0$. Also, with the introduction of $\mathbf{B}_k$, constraint $\mathrm{C3}$ is equivalent to
\begin{align}
    &\mathrm{\overline{C3}}: 1+{p_k\mathbf{h}}_k^H\mathbf{S}_c\mathbf{B}_k^{-1}\mathbf{S}_c\mathbf{h}_k\geq2^{\tau_k},\,\forall k,\nonumber \\
    &\mathrm{C6}: \mathbf{B}_k\!-\!\sigma_n^2\mathbf{I}\!=\!\mathbf{S}_c\!\!\left(\!\sum_{i\neq k}^{K}{p_i\mathbf{h}_i\mathbf{h}_i^H}\!\!+\!\!\sum_{m=1}^{M}{\mathbf{H}_m{\overline{\mathbf{V}}}_m\mathbf{H}_m^H}\!\right)\!\mathbf{S}_c, \,\forall k.
\end{align}
The left hand side (LHS) of $\mathrm{\overline{C3}}$ is jointly convex to $\mathbf{S}_c$ and $\mathbf{B}_k$. However, the equality constraint $\mathrm{C6}$ remains challenging.

The PDD method is a well-known and effective approach for handling such equality constraints and decoupling multivariate optimization \cite{pdd1}. Therefore, we employ the PDD method to address this challenging issue. Specifically, by introducing auxiliary variables $\{y_m^C
\}_{m=1}^M$ and $\left\{\overline{\mathbf{B}}_k\right\}_{k=1}^K$, we first recast problem (\ref{problem13}) as 
\begin{align}\label{problem16}
   & \mathop{\text{maximize}}_{\substack{\left\{x_m^C,y_m^C\right\},\left\{\mathbf{V}_m\succeq\mathbf{0}\right\},\\r,\left\{\tau_k,\chi_k\right\},\mathbf{B}_k}}\quad{r}\nonumber\\
   & \text{subject to}\quad \mathrm{\overline{C2}},\mathrm{C5},\mathrm{C6},\nonumber \\
   &\mathrm{\overline{C1}}: \delta\left(\mathrm{Tr}({\overline{\mathbf{V}}}_m)\right)+y_m^C\le1, \forall m, \nonumber\\
   &  \mathrm{\overline{C3}}:1+p_k\mathbf{h}_k^H\mathbf{S}_c\overline{\mathbf{B}}_k^{-1}\mathbf{S}_c\mathbf{h}_k\geq2^{\tau_k},\forall k\nonumber\\
   & \mathrm{\overline{C4}}:\sum_{m=1}^{M}{\frac{1}{p_k\left|\mathbf{w}_{e,k}^H\mathbf{g}_k\right|^2}\mathbf{w}_{e,k}^H\mathbf{G}_m{\overline{\mathbf{V}}}_m\mathbf{G}_m^H\mathbf{w}_{e,k}}\nonumber\\
   &\ \ \ \ \ \geq\!\frac{1}{2^{\chi_k}\!\!-\!\!1}\!\!-\!\!\frac{\sum_{i\neq k}{p_i\left|\mathbf{w}_{e,k}^H\mathbf{g}_i\right|^2}\!\!+\!\!\sigma_e^2\Vert\mathbf{w}_{e,k}\Vert^2}{p_k\left|\mathbf{w}_{e,k}^H\mathbf{g}_k\right|^2}, \forall k, \nonumber\\
   & \mathrm{C7}:x_m^C=y_m^C,\,x_m^C\in \{0,1\},\forall m, \nonumber\\
   & \mathrm{C8}:\mathbf{B}_k= \overline{\mathbf{B}}_k, \forall k.
\end{align}
The introduction of $\{y_m^C
\}_{m=1}^M$ and $\left\{\overline{\mathbf{B}}_k\right\}_{k=1}^K$ serves to decouple the variables within the constraints, thereby facilitating the subsequent alternating optimization. Then, the augmented Lagrangian problem for (\ref{problem16}) is formulated as
\begin{align}\label{aLP}
   \mathop{\text{maximize}}_{\substack{\left\{x_m^C,y_m^C\right\},\left\{\mathbf{V}_m\succeq\mathbf{0}\right\},\\r,\left\{\tau_k,\chi_k\right\},\mathbf{B}_k,\overline{\mathbf{B}}_k}}\quad&r-P_\rho\left(\{x_m^C,y_m^C\}_{m=1}^M,\left\{ \mathbf{B}_k,\overline{\mathbf{B}}_k\right\}_{k=1}^K\right)\nonumber\\
   \text{subject to}\quad & \mathrm{\overline{C1}},\mathrm{\overline{C2}},\mathrm{\overline{C3}},\mathrm{\overline{C4}},\mathrm{C5},
\end{align}
where $P_\rho\left(\{x_m^C,y_m^C\}_{m=1}^M,\left\{ \mathbf{B}_k,\overline{\mathbf{B}}_k\right\}_{k=1}^K\right)$ is defined in (\ref{al0}), $\mathbf{Z}_{1,k}$, $\mathbf{Z}_{2,k}$, and $z_{3,m}$ denote the dual variables associated with the constraints $\mathrm{C6}$, $\mathrm{C7}$, and $\mathrm{C8}$, respectively, and $\rho$ is the penalty parameter.
\begin{figure*}
\begin{align}\label{al0}
    P_\rho\left(\{x_m^C,y_m^C\}_{m=1}^M,\left\{ \mathbf{B}_k,\overline{\mathbf{B}}_k\right\}_{k=1}^K\right)\triangleq& \frac{1}{2\rho}\Bigg(\sum_{k=1}^K  \bigg \Vert\mathbf{B}_k-\sigma_n^2\mathbf{I}-\mathbf{S}_c\left(\sum_{i\neq k}^{K}{p_i\mathbf{h}_i\mathbf{h}_i^H}+\sum_{m=1}^{M}{\mathbf{H}_m{\overline{\mathbf{V}}}_m\mathbf{H}_m^H}\right)\mathbf{S}_c+\rho\mathbf{Z}_{1,k}\bigg\Vert_F^2 \Bigg.\nonumber \\
    &\Bigg. +\sum_{k=1}^K\left\Vert\mathbf{B}_k-\overline{\mathbf{B}}_k+\rho\mathbf{Z}_{2,k}\right\Vert_F^2+\sum_{m=1}^M\left(x_m^C-y_m^C+\rho z_{3,m}\right)^2\Bigg)
\end{align}
\hrulefill
\end{figure*}

\subsection{BCD Method for Solving Problem (\ref{aLP})}
Next, we employ the BCD method to solve problem (\ref{aLP}) in the inner iterations, while updating the dual variables and penalty parameter in the outer iterations. Specifically, the variables are divided into three blocks, and the corresponding update procedures in the BCD iterations are detailed as follows.

\subsubsection{Updating $\{y_m^C\}$, $\left\{\overline{\mathbf{V}}_m\succeq\mathbf{0}\right\}$, $r$ and $\left\{\chi_k\right\}$}
With the other variables fixed, the problem with respect to $\{y_m^C\}$, $\left\{\mathbf{V}_m\succeq\mathbf{0}\right\}$, $r$ and $\left\{\chi_k\right\}$ reduces to
\begin{align}\label{aLP1}
   \mathop{\text{maximize}}_{\left\{x_m^C\right\},\left\{\overline{\mathbf{V}}_m\succeq\mathbf{0}\right\},r,\left\{\chi_k\right\}}\quad&r-P_\rho\left(\{x_m^C,y_m^C\}_{m=1}^M,\left\{ \mathbf{B}_k,\overline{\mathbf{B}}_k\right\}_{k=1}^K\right)\nonumber\\
   \text{subject to}\quad & \mathrm{\overline{C1}},\mathrm{\overline{C2}},\mathrm{\overline{C4}},\mathrm{C5}.
\end{align}
Although the objective function is now convex, the constraints $\overline{\mathrm{C1}}$ and $\overline{\mathrm{C4}}$ remain nonconvex and need additional processing. For constraint $\overline{\mathrm{C1}}$, we first approximate $\delta\left(\mathrm{Tr}({\overline{\mathbf{V}}}_m)\right)$ by a concave function, i.e.,
\begin{align}
    \delta\left(\mathrm{Tr}({\overline{\mathbf{V}}}_m)\right)\approx \frac{2}{\pi}\mathrm{arctan}\left(\frac{\mathrm{Tr}({\overline{\mathbf{V}}}_m}{\varpi}\right)\triangleq f_\varpi (\mathrm{Tr}({\overline{\mathbf{V}}}_m)),
\end{align}
where $\varpi> 0$ is a parameter controlling the approximation accuracy, with smaller values yielding higher accuracy.

Now, we are ready to adopt the SCA technique to handle the nonconvex constraints $\overline{\mathrm{C1}}$ and $\overline{\mathrm{C4}}$. In specific, we consider the following upper bound for the concave parts:
\begin{align}\label{eq21}
    &f_\varpi (\mathrm{Tr}({\overline{\mathbf{V}}}_m))\leq f_\varpi \left(\mathrm{Tr}\left({\overline{\mathbf{V}}}_m^{(t)}\right)\right)\nonumber \\
    &\quad \quad \quad \quad \quad \quad \quad +\nabla f_\varpi \left(\mathrm{Tr}\left({\overline{\mathbf{V}}}_m^{(t)}\right)\right) \left( \mathrm{Tr}\left({\overline{\mathbf{V}}}_m-{\overline{\mathbf{V}}}_m^{(t)}\right)\right),\nonumber\\
    &\frac{1}{2^{\chi_k}-1}\leq \frac{1}{2^{\chi_k^{(t)}}-1}-\frac{\mathrm{ln}2 \cdot 2^{\chi_k^{(t)}}}{\left(2^{\chi_k^{(t)}}-1\right)^2} \left(\chi_k -\chi_k^{(t)}\right),
\end{align}
where ${\overline{\mathbf{V}}}_m^{(t)}$ and $\chi_k^{(t)}$ are the solutions in the $t$-th iteration. Based on the convex approximations in (\ref{eq21}), we replace the concave parts in  $\overline{\mathrm{C1}}$ and $\overline{\mathrm{C4}}$ and the new constraints are denoted by  $\widetilde{\mathrm{C1}}$ and $\widetilde{\mathrm{C4}}$. In the $(t+1)$-th iteration, we consider the following convex surrogate problem
\begin{align}\label{eq22}
   \mathop{\text{maximize}}_{\left\{x_m^C\right\},\left\{\mathbf{V}_m\succeq\mathbf{0}\right\},r,\left\{\chi_k\right\}}\quad&r-P_\rho\left(\{x_m^C,y_m^C\}_{m=1}^M,\left\{ \mathbf{B}_k,\overline{\mathbf{B}}_k\right\}_{k=1}^K\right)\nonumber\\
   \text{subject to}\quad & \mathrm{\widetilde{C1}},\mathrm{\overline{C2}},\mathrm{\widetilde{C4}},\mathrm{C5},
\end{align}
which can be efficiently solved by existing numerical convex program solvers \cite{cvx}. After solving (\ref{eq22}), we update ${\overline{\mathbf{V}}}_m^{(t)}$ and $\chi_k^{(t)}$ using the obtained optimal solutions, and then proceed to the next iteration. It is worth noting that the iterative procedure is guaranteed to converge to a Karush–Kuhn–Tucker (KKT) point of problem (\ref{aLP1}) \cite{beck2010sequential,yaojsac}.

\subsubsection{Updating $\mathbf{B}_k$}
The subproblem for optimizing $\mathbf{B}_k$ is given as (\ref{aLP2}),
\begin{figure*}
\begin{align}\label{aLP2}
   \mathop{\text{minimize}}_{\mathbf{B}_k}\quad&\bigg\Vert\mathbf{B}_k-\sigma_n^2\mathbf{I}-\mathbf{S}_c\left(\sum_{i\neq k}^{K}{p_i\mathbf{h}_i\mathbf{h}_i^H}+\sum_{m=1}^{M}{\mathbf{H}_m{\overline{\mathbf{V}}}_m\mathbf{H}_m^H}\right)\mathbf{S}_c+\rho\mathbf{Z}_{1,k}\bigg\Vert_F^2+\big\Vert\mathbf{B}_k-\overline{\mathbf{B}}_k+\rho\mathbf{Z}_{2,k}\big\Vert_F^2
\end{align}
\hrulefill
\end{figure*}
which is convex and we can directly derive its optimal solution. Specifically, it follows
\begin{align}\label{eq24}
   \mathbf{B}_k=&\frac{1}{2}\Bigg(\sigma_n^2\mathbf{I}+\mathbf{S}_c\left(\sum_{i\neq k}^{K}{p_i\mathbf{h}_i\mathbf{h}_i^H}+\sum_{m=1}^{M}{\mathbf{H}_m{\overline{\mathbf{V}}}_m\mathbf{H}_m^H}\right)\mathbf{S}_c\Bigg.\nonumber\\
   &\Bigg.-\rho\mathbf{Z}_{1,k}+\overline{\mathbf{B}}_k-\rho\mathbf{Z}_{2,k}\Bigg).
\end{align}

\subsubsection{Updating $\left\{x_m^C\right\}$, $r$, $\left\{\tau_k\right\}$, and $\overline{\mathbf{B}}_k$}
Fixing the other variables, this subproblem is given as (\ref{aLP3}).
\begin{figure*}
\begin{align}\label{aLP3}
   \mathop{\text{maximize}}_{\substack{\left\{x_m^C\right\},r,\left\{\tau_k\right\},\overline{\mathbf{B}}_k}}\quad&r-\frac{1}{2\rho}\Bigg(\sum_{k=1}^K \bigg\Vert\mathbf{B}_k-\sigma_n^2\mathbf{I}-\mathbf{S}_c\left(\sum_{i\neq k}^{K}{p_i\mathbf{h}_i\mathbf{h}_i^H}+\sum_{m=1}^{M}{\mathbf{H}_m{\overline{\mathbf{V}}}_m\mathbf{H}_m^H}\right)\mathbf{S}_c+\rho\mathbf{Z}_{1,k}\bigg\Vert_F^2\Bigg.\nonumber\\
   &\ \ \ \ \ \ \ \ \ \ \ \ \ \ \ +\sum_{k=1}^K\Bigg.\big\Vert\mathbf{B}_k-\overline{\mathbf{B}}_k+\rho\mathbf{Z}_{2,k}\big\Vert_F^2+\sum_{m=1}^M\left(x_m^C-y_m^C+\rho z_{3,m}\right)^2\Bigg)\nonumber\\
   \text{subject to}\quad & \mathrm{\overline{C3}},\mathrm{C5}.
\end{align}
\hrulefill
\end{figure*}
Considering that the objective involves the fourth order term of $x_m^C$, we first simplify it as follows:
\begin{align}
    &\sum_{k=1}^{K} \Vert\mathbf{S}_c\mathbf{Q}_k\mathbf{S}_c\!\!-\!\!\mathbf{C}_k\Vert_F^2\nonumber \\
    & = \!\!\sum_{k=1}^{K} \!\mathrm{Tr}\!\left(\mathbf{S}_c\mathbf{Q}_k\mathbf{S}_c\mathbf{S}_c\mathbf{Q}_k\mathbf{S}_c\right)\!-\!2\mathfrak{R}\!\left\{\!\mathrm{Tr}\!\left(\mathbf{S}_c\mathbf{Q}_k\mathbf{S}_c\mathbf{C}_k^H\right)\!\right\}\!+\!\Vert\mathbf{C}_k\Vert_F^2\nonumber \\
    &=\!\sum_{k=1}^{K}{\mathrm{Tr}\left(\mathbf{S}_c\mathbf{Q}_k\mathbf{S}_c\mathbf{Q}_k\right)-2\mathfrak{R}\left\{\mathrm{Tr}\left(\mathbf{S}_c\mathbf{Q}_k\mathbf{S}_c\mathbf{C}_k^H\right)\right\}}\nonumber\\
    &=\!\sum_{k=1}^{K}\mathfrak{R}\left\{\mathrm{Tr}\left(\mathbf{S}_c\mathbf{Q}_k\mathbf{S}_c(\mathbf{Q}_k-2\mathbf{C}_k^H\right)\right\},
\end{align}
where $\mathbf{Q}_k=\sum_{i\neq k}^{K}{p_i\mathbf{h}_i\mathbf{h}_i^H}+\sum_{m=1}^{M}{\mathbf{H}_m{\overline{\mathbf{V}}}_m\mathbf{H}_m^H}$, $\mathbf{C}_k=\mathbf{B}_k-\sigma_n^2\mathbf{I}+\rho\mathbf{Z}_{1,k}$, and the equality exploits the fact that $\mathbf{S}_c^2=\mathbf{S}_c$. Then, by applying \cite[Eq. (1.10.6)]{xdzhang}, we have 
\begin{align} \label{eq27}
    &\sum_{k=1}^{K}\mathfrak{R}\left\{\mathrm{Tr}\left(\mathbf{S}_c\mathbf{Q}_k\mathbf{S}_c(\mathbf{Q}_k-2\mathbf{C}_k^H\right)\right\}\nonumber\\
    &\quad=\mathbf{x}^H\left(\sum_{k=1}^{K}{\mathbf{Q}_k\odot\left(\mathbf{Q}_k-2\mathbf{C}_k^H\right)}\right)\mathbf{x},
\end{align}
where $\mathbf{x}\triangleq \left[x_1^C\mathbf{1};\cdots;x_M^C\mathbf{1}\right]$. Since the positive definiteness of $\sum_{k=1}^{K}{\mathbf{Q}_k\odot\left(\mathbf{Q}_k-2\mathbf{C}_k^H\right)}$ cannot be determined, it is unable to characterize the convexity or concavity of (\ref{eq27}). To handle this issue, we propose an MM-based method. According to \cite[Eq. (26)]{mm}, we employ the following upper bound to the term in (\ref{eq27}):
\begin{align}
    &\mathbf{x}^H\left(\sum_{k=1}^{K}{\mathbf{Q}_k\odot\left(\mathbf{Q}_k-2\mathbf{C}_k^H\right)}\right)\mathbf{x} \leq \mathbf{x}^H\mathbf{M}\mathbf{x}\nonumber\\
    &+2\Re\left\{\mathbf{x}^H \left(\left(\sum_{k=1}^{K}{\mathbf{Q}_k\odot\left(\mathbf{Q}_k-2\mathbf{C}_k^H\right)}\right)-\mathbf{M}\right)\mathbf{x}^{(t)}\right\}\nonumber \\
    &+\!\!\left(\mathbf{x}^{(t)}\right)^H \!\!\!\left(\!\mathbf{M}\!-\!\!\left(\sum_{k=1}^{K}{\mathbf{Q}_k\!\odot\!\left(\mathbf{Q}_k\!-\!2\mathbf{C}_k^H\right)}\!\!\right)\!\!\right)\mathbf{x}^{(t)}\!\triangleq\! g\!\left(\mathbf{x}|\mathbf{x}^{(t)}\right),
\end{align}
where $\mathbf{x}^{(t)}\!=\!\left[x_1^{C,(t)}\mathbf{1};\cdots;x_M^{C,(t)}\mathbf{1}\right]$ is the optimal solution obtained in the previous iteration, $\sum_{k=1}^{K}{\mathbf{Q}_k\odot\left(\mathbf{Q}_k-2\mathbf{C}_k^H\right)}=\mathbf{U}\mathbf{S}\mathbf{U}^H$ and $\mathbf{M}=\mathbf{U}\max\{\mathbf{S},\mathbf{0}\}\mathbf{U}^H$. Moreover, we note that $\mathrm{\overline{C3}}$ is still nonconvex with respect to $\overline{\mathbf{B}}_k$. According to \cite[Eq. (14)]{mm}, the LHS can be lower bounded by
\begin{align}
    &p_k\mathbf{h}_k^H\mathbf{S}_c \overline{\mathbf{B}}_k^{-1}\mathbf{S}_c\mathbf{h}_k \geq   2p_k\mathrm{Tr}\left( \mathbf{S}_c^{(t)} \left(\overline{\mathbf{B}}_k^{(t)}\right)^{-1}\mathbf{S}_c\mathbf{h}_k \mathbf{h}_k^H\right)\nonumber \\
    &\quad -\mathrm{Tr}\!\left( \!\! \left(\overline{\mathbf{B}}_k^{(t)}\right)^{-1}\!\!\mathbf{S}_c^{(t)}\mathbf{h}_k \mathbf{h}_k^H\mathbf{S}_c^{(t)}\left(\overline{\mathbf{B}}_k^{(t)}\right)^{-1}\! \overline{\mathbf{B}}_k\right)\nonumber \\
    &\triangleq\! h\left(\overline{\mathbf{B}}_k,\mathbf{S}_c|\overline{\mathbf{B}}_k^{(t)},\mathbf{S}_c^{(t)}\right),
\end{align}
where $\overline{\mathbf{B}}_k^{(t)}$ is the optimal solution to $\overline{\mathbf{B}}_k$ in the $t$-th iteration. Then, regarding the discrete constraint of $x_m^C$, we equivalently transform it into
\begin{align}
    &{\mathrm{C9}}: (x_m^C)^2-x_m^C\leq 1, \forall m,\nonumber \\
    &{\mathrm{C10}}: (x_m^C)^2-x_m^C\geq 1, \forall m.
\end{align}
The nonconvex constraint $\mathrm{C10}$ can be approximated by 
\begin{align}
    \overline{\mathrm{C10}}: 2 x_m^{C,(t)} x_m^C- (x_m^{C,(t)})^2-x_m^C\geq 1, \forall m.
\end{align}

Now, we construct the convex surrogate problem at the $t$-th iteration as
\begin{align}\label{eq32}
   \mathop{\text{maximize}}_{\substack{\left\{x_m^C\right\},r,\left\{\tau_k\right\},\overline{\mathbf{B}}_k}}\enspace&r\!-\!\frac{1}{2\rho}g\!\left(\mathbf{x}|\mathbf{x}^{(t)}\right)\!\!-\!\frac{1}{2\rho}\!\! \Bigg(\!\!\sum_{k=1}^K\! \big\Vert\mathbf{B}_k\!-\!\overline{\mathbf{B}}_k\!+\!\rho\mathbf{Z}_{2,k}\big\Vert_F^2\Bigg. \nonumber \\
   &\Bigg.+\sum_{m=1}^M\left(x_m^C-y_m^C+\rho z_{3,m}\right)^2\Bigg)\nonumber\\
   \text{subject to}\quad & \mathrm{\widetilde{C3}}:1+h\left(\overline{\mathbf{B}}_k,\mathbf{S}_c|\overline{\mathbf{B}}_k^{(t)},\mathbf{S}_c^{(t)}\right) \geq 2^{\tau_k},\forall k\nonumber \\
   &\mathrm{C5}, \overline{\mathrm{C10}}.
\end{align}
The above process is iterated until convergence, which guarantees that the solution ultimately converges to a stationary point of problem (\ref{aLP3}).

Furthermore, in the outer iteration, the dual variables are updated by 
\begin{align}\label{dual}
    &\mathbf{Z}_{1,k} = \mathbf{Z}_{1,k} +\frac{1}{\rho}\left(\mathbf{B}_k -\sigma_n^2 \mathbf{I}-\mathbf{S}_c\mathbf{Q}_k \mathbf{S}_c\right),\forall k,\nonumber \\
    &\mathbf{Z}_{2,k}=\mathbf{Z}_{2,k}+\frac{1}{\rho}\left(\mathbf{B}_k-\overline{\mathbf{B}}_k\right),\forall k, \nonumber \\
    &z_{3,m} =z_{3,m}+\frac{1}{\rho}(x_m^C-y_m^C),\forall k.
\end{align}

The detailed steps of the proposed PDD method for solving problem~(\ref{problem13}) are listed in Algorithm~\ref{alg1}, where $c$ is a constant for decreasing the penalty parameter. According to \cite{pdd1}, the ultimate solution obtained via Algorithm~\ref{alg1} converges to a stationary point of problem (\ref{problem13}). Within the PDD framework, the primal variables are updated via BCD iterations, while the dual variables and penalty parameters are sequentially refined. The dominant computational burden arises from updating $\overline{\mathbf{B}}_k$ and $\overline{\mathbf{V}}_m$, with a complexity on the order of $\mathcal{O}(LK^{3.5}M^{6.5}N_a^{6.5})$, where $L$ denotes the SCA iteration number.

\begin{algorithm}[!t]
\caption{Proposed algorithm for solving (\ref{problem13})} \label{alg1}
\begin{algorithmic}[1]  
\STATE \textbf{Initialize} feasible initial variables, dual variables $\mathbf{Z}_{1,k}^{(0)}$, $\mathbf{Z}_{2,k}^{(0)}$, $z_{3,m}^{(0)}$, penalty parameter $\rho^{(0)}$, constants $c \in(0,1)$ and $\ell=0$.
\REPEAT
\REPEAT
\STATE \textbf{Initialize} $t=0$ and initial variables ${\overline{\mathbf{V}}}_m^{(0)}$ and $\chi_k^{(0)}$.
\REPEAT
\STATE Solving (\ref{eq22}) and update ${\overline{\mathbf{V}}}_m^{(t)}$ and $\chi_k^{(t)}$.
\STATE Set $t=t+1$.
\UNTIL convergence.
\STATE Update $\mathbf{B}_k$, $\forall k$  according to (\ref{eq24}).
\STATE \textbf{Initialize} $t=0$ and initial variables $x_m^{C,(0)}$ and $\overline{\mathbf{B}}_k^{(0)}$.
\REPEAT
\STATE Solving (\ref{eq32}) and update $x_m^{C,(t)}$ and $\overline{\mathbf{B}}_k^{(t)}$.
\STATE Set $t=t+1$.
\UNTIL convergence.
\UNTIL convergence.
\STATE Update the dual variables by (\ref{dual}).
\STATE Update the penalty parameter as $\rho^{(\ell+1)}=c\rho^{(\ell)}$.
\STATE Set $\ell\gets\ell+1$.
\UNTIL convergence.
\end{algorithmic} 
\end{algorithm}

\section{Low-Complexity Algorithm}
To alleviate the computational burden of jointly optimizing AP modes and AN covariance, a sequential low-complexity  algorithm is developed in this section. Specifically, AP modes are first determined via a heuristic rule, followed by a closed-form iterative procedure for AN optimization.
 
\subsection{Heuristic AP Mode Selection}
To reduce the complexity associated with the joint optimization of AP modes and AN covariance matrices, we develop a low-complexity heuristic that determines the AP modes in a sequential manner. The key idea is that, when selecting the $n$-th R-AP, only the unassigned APs are considered as potential T-APs. Consequently, the interference vulnerability of each candidate R-AP is evaluated based on the actual set of APs that may transmit AN at that stage, rather than assuming that all APs can act as jammers simultaneously. This sequential evaluation avoids interference overestimation and leads to more accurate and adaptive mode decisions.

\subsubsection{Metric Definitions}

Let $\mathcal{R}$ denote the set of APs already selected as R-APs, and let $\mathcal{C}$ represent the set of unassigned APs, i.e., $\mathcal{C}=\{1,\ldots,M\}\setminus\mathcal{R}$. For each AP $m\in\mathcal{C}$, the following four metrics are introduced to assess its suitability for reception or jamming.

\paragraph{Reception gain (RG)}
The uplink reception quality of AP $m$ is measured as
\begin{align}
    \mathrm{RG}_m\triangleq\sum_{k=1}^{K} p_k \Vert \mathbf{h}_{m,k} \Vert^2,
\end{align}
capturing its aggregated signal strength from all UAVs.

\paragraph{Vulnerability to potential T-AP interference}
At the current selection step, only APs in $\mathcal{C}\setminus\{m\}$ may operate as T-APs. Thus, the AN leakage at AP $m$ is approximated by
\begin{align}
    \mathrm{VUL}_m(\mathcal{C})
    \triangleq 
    \sum_{j\in\mathcal{C}\setminus\{m\}}P_j\Vert\mathbf{H}_{m,j}\Vert_F^2,
\end{align}
which accurately reflects the interference risk under the sequential mode-selection procedure.

\paragraph{Jamming gain (JG)}
The capability of AP $m$ to jam Eve is captured by
\begin{align}
    \mathrm{JG}_m 
    \triangleq P_m
    \Vert \mathbf{G}_m \Vert_F^2.
\end{align}
A larger value indicates that AP $m$ is better suited to function as a T-AP.

\paragraph{Backfire to existing R-APs}
If AP $m$ becomes a T-AP later, its AN may cause interference to the already selected R-APs as
\begin{align}
    \mathrm{BF}_m(\mathcal{R})
    \triangleq
    \sum_{n \in \mathcal{R}} 
        P_m\Vert \mathbf{H}_{n,m} \Vert_F^2.
\end{align}

\subsubsection{Sequential Priority Score}
Based on the above metrics, the priority score for each candidate $m\in\mathcal{C}$ is defined as
\begin{align}\label{eq:score}
    \mathrm{score}_m(\mathcal{C})
    \triangleq
    \frac{\mathrm{RG}_m}
         {\,\mathrm{VUL}_m(\mathcal{C}) + N_a\sigma_n^2\,}\times\frac{1}{1+\beta\frac{\mathrm{JG}_m}{\mathrm{BF}_m(\mathcal{R})+\epsilon}},
\end{align}
where $\beta\ge 0$ adjusts the preference for preserving strong jammers as T-APs, and $\epsilon>0$ guarantees numerical stability.
The first fraction favors APs with stronger reception capability and lower susceptibility to AN interference, whereas the second fraction penalizes APs that either have high jamming potential (desirable to keep for T-AP roles) or would introduce strong AN leakage to already selected R-APs. 
Similar AP selection strategies have been explored in scalable CF systems \cite{bjornson2020scalable}, supporting the feasibility of heuristic sequential assignment to balance performance and complexity. Our sequential-priority scoring can thus be seen as a natural extension tailored to secrecy-oriented flexible-duplex operation.

\subsubsection{Greedy Assignment Procedure}
Let $N_C = \lfloor \theta M \rfloor$ be the required number of R-APs. The heuristic begins with $\mathcal{R}=\emptyset$ and $\mathcal{C}=\{1,\ldots,M\}$, and iteratively selects R-APs as
\begin{align}
    m^\star = \arg\max_{m\in\mathcal{C}} 
                \mathrm{score}_m(\mathcal{C}).
\end{align}
Then, we update the sets as follows
\begin{align}
    \mathcal{R} \leftarrow \mathcal{R} \cup \{m^\star\},\qquad
    \mathcal{C} \leftarrow \mathcal{C} \setminus \{m^\star\}.
\end{align}

After selecting $N_C$ R-APs, the final AP mode assignment is given by the following piecewise rule
\begin{align}
(x_m^C,x_m^J)=
\begin{cases}
(1,0), & m\in\mathcal{R},\\
(0,1), & m\in\mathcal{C}.
\end{cases}
\end{align}
This heuristic sequential assignment captures the asymmetric roles of R-APs and T-APs in the flexible-duplex CF architecture, and provides an efficient initialization for subsequent AN covariance optimization.



\subsection{AN Optimization with Fixed AP Allocation}
Given a fixed AP allocation, we proceed to optimize the AN covariance matrices. Before that, we update the effective channel vector $\widetilde{\mathbf{h}}_k\in \mathbb{C}^{N N_a\times 1}$ by removing the zero component (i.e., the channel corresponding to the $m$-th AP with $x_m^J=1$) from the original channel vector $\mathbf{h}_k$, where $N\!=\!\sum_{m=1}^M x_m^C$ is the number of R-APs. Similarly, the effective inter-AP channel matrix is updated as $\widetilde{\mathbf{H}}_m\!\in\! \mathbb{C}^{N N_a\!\times\! N_a}$. With the effective channels, we rewrite the AN optimization problem as
\begin{align}\label{ANopt}
   \mathop{\text{maximize}}_{\left\{\mathbf{V}_m\succeq\mathbf{0}\right\}_{m\in\mathcal{A}}}\enspace&{\min_{k}\frac{1\!+\!p_k{\widetilde{\mathbf{h}}}_k^H\left(\mathbf{D}_k\!+\!\sum_{m\in\mathcal{A}}\!{{\widetilde{\mathbf{H}}}_m\mathbf{V}_m{\widetilde{\mathbf{H}}}_m^H}\right)^{\!-1}{\widetilde{\mathbf{h}}}_k}{1+\frac{1}{c_k+\frac{1}{{\widetilde{p}}_k}\sum_{m\in\mathcal{A}}{\mathbf{w}_{e,k}^H\mathbf{G}_m\mathbf{V}_m\mathbf{G}_m^H\mathbf{w}_{e,k}}}}}\nonumber\\
   \text{subject to}\quad & \mathrm{Tr}\left(\mathbf{V}_m\right)\le P_m, \forall m\in\mathcal{A},
\end{align}
where $\mathcal{A}$ denotes the set of T-APs, $\mathbf{D}_k\triangleq \sum_{i\neq k}^{K}{p_i\widetilde{\mathbf{h}}_i\widetilde{\mathbf{h}}_i^H}+\sigma_n^2\mathbf{I}_{N N_a}$, $c_k\triangleq \frac{\sum_{i\neq k} p_i |\mathbf{w}_{e,k}^H \mathbf{g}_i|^2+\sigma_e^2 \Vert \mathbf{w}_{e,k}\Vert^2}{p_k |\mathbf{w}_{e,k}^H \mathbf{g}_k|^2} $, and $\tilde{p}_k \triangleq \frac{1}{p_k |\mathbf{w}_{e,k}^H \mathbf{g}_k|^2} $. This problem is also highly nonconvex. To enable efficient and rapid solving, we continue to employ the PDD framework for variable decoupling and problem decomposition.

To be concrete, we equivalently rewrite problem (\ref{ANopt}) as
\begin{align}\label{eq36}
   \mathop{\text{maximize}}\limits_{\substack{\left\{\!\mathbf{V}_m,{\widetilde{\mathbf{V}}}_m\right\}_{m\in\mathcal{A}},\\
   \left\{\tau_k,\chi_k,u_k,\mathbf{Q}_k\right\}}}&\enspace{\min_{k}{\tau_k-u_k}}\nonumber\\
   \text{subject to}\enspace & \mathrm{Tr}\left({\widetilde{\mathbf{V}}}_m\right)\le x_m^JP_m, \forall m\in\mathcal{A},\nonumber\\
   & {\widetilde{\mathbf{V}}}_m\succcurlyeq\mathbf{0},\forall m\in\mathcal{A},\nonumber\\
   & 1+p_k{\widetilde{\mathbf{h}}}_k^H\mathbf{Q}_k{\widetilde{\mathbf{h}}}_k=2^{\tau_k},\forall k,\nonumber\\
   & c_k\!+\!\frac{1}{{\widetilde{p}}_k}\!\!\sum_{m\in\mathcal{A}}{\mathbf{w}_{e,k}^H\mathbf{G}_m\mathbf{V}_m\mathbf{G}_m^H\mathbf{w}_{e,k}}\!=\!\frac{1}{2^{\chi_k}\!-\!1},\forall k,\nonumber\\
   & \mathbf{Q}_k\left(\mathbf{D}_k\!+\!\sum_{m\in\mathcal{A}}\!{{\widetilde{\mathbf{H}}}_m\mathbf{V}_m{\widetilde{\mathbf{H}}}_m^H}\right)=\mathbf{I},\forall k,\nonumber\\
   & \mathbf{V}_m={\widetilde{\mathbf{V}}}_m,\forall m\in\mathcal{A},\nonumber\\
   & u_k=\chi_k,\forall k,
\end{align}
where $\tau_k,\chi_k,u_k,\mathbf{Q}_k$, and ${\widetilde{\mathbf{V}}}_m$ are newly introduced auxiliary variables. Among them, $\tau_k$ and $\chi_k$ are introduced to decompose the complicated fractional terms,  and $\mathbf{Q}_k$ is introduced to substitute the matrix inverse. In addition, the introduction of $\widetilde{\mathbf{V}}_m$ and $u_k$ serves as auxiliary replicas of $\mathbf{V}_m$ and $\chi_k$, which facilitates deriving closed-form optimal solutions for the corresponding subproblems. Next, the augmented Lagrange problem can be formulated as
\begin{align}
   \mathop{\text{maximize}}_{\mathcal{X}}&\enspace \min_k{\ \tau_k-\chi_k}-P_\rho(\mathcal{X})\nonumber\\
   \text{subject to}\enspace & \mathrm{Tr}\left({\widetilde{\mathbf{V}}}_m\right)\le x_m^JP_m, \forall m\in\mathcal{A},\nonumber\\
   & {\widetilde{\mathbf{V}}}_m\succcurlyeq\mathbf{0},\forall m\in\mathcal{A},
\end{align}
where $\mathcal{X}\triangleq\left\{ \left\{\!\mathbf{V}_m,{\widetilde{\mathbf{V}}}_m\right\}_{m\in\mathcal{A}},
   \left\{\tau_k,\chi_k,u_k,\mathbf{Q}_k\right\}_{k=1}^K \right\}$ and
\begin{align}
    P_\rho (\mathcal{X})\!\triangleq \! &\frac{1}{2\rho}\sum_{k=1}^{K}\left(\left(1+p_k{\widetilde{\mathbf{h}}}_k^H\mathbf{Q}_k{\widetilde{\mathbf{h}}}_k-2^{\tau_k}+\rho z_{1,k}\right)^2\right.\nonumber\\
    &+\!\!\!\left.\left( \!\!c_k\!+\!\frac{\sum_{m\in\mathcal{A}}\!\mathbf{w}_{e,k}^H\mathbf{G}_m\!\!\mathbf{V}_m\!\mathbf{G}_m^H\mathbf{w}_{e,k}}{{\widetilde{p}}_k}\!\!-\!\!\frac{1}{2^{\chi_k}\!\!-\!\!1}\!\!+\!\!\rho z_{2,k}\!\!\right)^2\right.\nonumber\\
   &+\left.\left \Vert\mathbf{Q}_k\left(\mathbf{D}_k\!+\!\sum_{m\in\mathcal{A}}\!{{\widetilde{\mathbf{H}}}_m\mathbf{V}_m{\widetilde{\mathbf{H}}}_m^H}\right)-\mathbf{I}+\rho\mathbf{Z}_{3,k}\right \Vert_F^2\right.\nonumber\\
   &+\left.\left(u_k-\chi_k+\rho z_{5,k}\right)^2\right)\nonumber\\
   &+\frac{1}{2\rho}\sum_{m=1}^{M}\Vert\mathbf{V}_m-{\widetilde{\mathbf{V}}}_m+\rho\mathbf{Z}_{4,m}\Vert_F^2.
\end{align}
Then, BCD method is adopted for updating variables in the inner iterations. 

\subsubsection{Updating $\mathbf{Q}_k$} The subproblem regarding $\mathbf{Q}_k$ is 
\begin{align}
    &\mathop{\text{minimize}}_{\mathbf{Q}_k} \quad\left(1+p_k{\widetilde{\mathbf{h}}}_k^H\mathbf{Q}_k{\widetilde{\mathbf{h}}}_k-2^{\tau_k}+\rho z_{1,k}\right)^2\nonumber\\
    &\quad +\left \Vert\mathbf{Q}_k\left(\mathbf{D}_k\!+\!\sum_{m\in\mathcal{A}}\!{{\widetilde{\mathbf{H}}}_m\mathbf{V}_m{\widetilde{\mathbf{H}}}_m^H}\right)-\mathbf{I}+\rho\mathbf{Z}_{3,k}\right \Vert_F^2,
\end{align}
which is a unconstrained convex problem. To further reformulate the objective, we vectorize the matrix to obtain
\begin{align}
&{\widetilde{\mathbf{h}}}_k^H\mathbf{Q}_k{\widetilde{\mathbf{h}}}_k\!=\!\left(\!\mathrm{vec}\!\left({\widetilde{\mathbf{h}}}_k{\widetilde{\mathbf{h}}}_k^H\!\right)\!\right)^T\!\!\mathrm{vec}\left(\mathbf{Q}_k\right)\!=\!\mathbf{a}_k^H\mathbf{q}_k,\nonumber \\
&\left \Vert\mathbf{Q}_k\left(\mathbf{D}_k\!+\!\sum_{m\in\mathcal{A}}\!{{\widetilde{\mathbf{H}}}_m\mathbf{V}_m{\widetilde{\mathbf{H}}}_m^H}\right)-\mathbf{I}+\rho\mathbf{Z}_{3,k}\right \Vert_F^2\nonumber\\
&=\Vert\mathbf{B}_k\mathbf{q}_k\!-\!\mathbf{b}_k\Vert^2,
\end{align}
where $\mathbf{a}_k\triangleq \left(\!\mathrm{vec}\!\left({\widetilde{\mathbf{h}}}_k{\widetilde{\mathbf{h}}}_k^H\!\right)\!\right)^T$, $\mathbf{q}_k \triangleq \mathrm{vec}(\mathbf{Q}_k)$, $\mathbf{B}_k\triangleq \left(\mathbf{D}_k\!+\!\sum_{m\in\mathcal{A}}\!{{\widetilde{\mathbf{H}}}_m\mathbf{V}_m{\widetilde{\mathbf{H}}}_m^H}\right)^T\otimes \mathbf{I}$ and $\mathbf{b}_k \triangleq \mathrm{vec}\left(\mathbf{I}\!-\!\rho\mathbf{Z}_{3,k}\right)$. Now,
the objective function is equivalent to
\begin{align}
&\Big(p_k{\widetilde{\mathbf{h}}}_k^H\mathbf{Q}_k{\widetilde{\mathbf{h}}}_k-\underbrace{(2^{\tau_k}-1-\rho z_{1,k})}_{\triangleq a_k}\Big)^2+\Vert\mathbf{B}_k\mathbf{q}_k\!-\!\mathbf{b}_k\Vert^2\nonumber\\
&\!\Leftrightarrow \!\mathbf{q}_k^H\!\left(p_k^2\mathbf{a}_k\mathbf{a}_k^H\!\!+\!\mathbf{B}_k^H\mathbf{B}_k\right)\mathbf{q}_k\!-\!2\Re\!\left\{\!\left(a_kp_k\mathbf{a}_k\!+\!\!\mathbf{B}_k^H\mathbf{b}_k\right)^H\! \!\mathbf{q}_k\!\right\}.
\end{align}
Hence, the optimal solution is derived as
\begin{align}\label{eq43}
\mathbf{Q}_k^\ast=\mathrm{mat}\!\left(\!\left(p_k^2\mathbf{a}_k\mathbf{a}_k^H\!\!+\!\!\mathbf{B}_k^H\mathbf{B}_k\right)^{\dag}\left(a_kp_k\mathbf{a}_k\!\!+\!\!\mathbf{B}_k^H\mathbf{b}_k\right)\!\right).
\end{align}

\subsubsection{Updating $\mathbf{V}_m$}  Next, we focus our attention on optimizing $\mathbf{V}_m $. The corresponding subproblem is formulated as
\begin{align}
    &\mathop{\text{minimize}}_{\mathbf{V}_m}\enspace\sum_{k=1}^{K}\!\left(\Vert\mathbf{Q}_k\left(\mathbf{D}_k\!+\!\sum_{m\in\mathcal{A}}\!{{\widetilde{\mathbf{H}}}_m\mathbf{V}_m{\widetilde{\mathbf{H}}}_m^H}\right)\!-\!\mathbf{I}\!+\!\rho\mathbf{Z}_{3,k}\Vert^2\right. \nonumber \\
    &+\!\left.\left(\!c_k\!+\!\frac{1}{{\widetilde{p}}_k}\sum_{m\in\mathcal{A}}{\mathbf{w}_{e,k}^H\mathbf{G}_m\mathbf{V}_m\mathbf{G}_m^H\mathbf{w}_{e,k}}\!-\!\frac{1}{2^{\chi_k}\!-\!1}\!+\!\rho z_{2,k}\!\right)^2\right)\nonumber\\
    &+\Vert\mathbf{V}_m-{\widetilde{\mathbf{V}}}_m+\rho\mathbf{Z}_{4,m}\Vert^2.
\end{align}
Similar to the updating procedure of $\mathbf{Q}_k$, we perform the following reformulations:
\begin{align}
    &\left(c_k\!+\!\frac{1}{{\widetilde{p}}_k}\sum_{m\in\mathcal{A}}{\mathbf{w}_{e,k}^H\mathbf{G}_m\mathbf{V}_m\mathbf{G}_m^H\mathbf{w}_{e,k}}\!-\!\frac{1}{2^{\chi_k}-1}\!+\!\rho z_{2,k}\right)^2\nonumber\\
    =&\Bigg(\mathrm{Tr}\Bigg(\mathbf{V}_m\underbrace{\frac{1}{{\widetilde{p}}_k}\mathbf{G}_m^H\mathbf{w}_{e,k}\mathbf{w}_{e,k}^H\mathbf{G}_m}_{\triangleq\mathbf{\Upsilon}_{m,k}}\Bigg)+\Bigg.\nonumber\\
    &\Bigg.\underbrace{c_k\!+\!\frac{1}{{\widetilde{p}}_k}\!\sum_{m^\prime\neq m}{\mathbf{w}_{e,k}^H\mathbf{G}_{m^\prime}\mathbf{V}_{m^\prime}\mathbf{G}_{m^\prime}^H\mathbf{w}_{e,k}}\!-\!\frac{1}{2^{\chi_k}\!-\!1}\!+\!\rho z_{2,k}}_{\triangleq \zeta_{m,k}}\Bigg)^2\nonumber\\
    =&\left(\left(\mathrm{vec}\left(\mathbf{\Upsilon}_{m,k}\right)\right)^T\mathrm{vec}\left(\mathbf{V}_m\right)+\zeta_{m,k}\right)^2\nonumber\\
    =&\mathbf{v}_m^H\bm{\upsilon}_{m,k}\bm{\upsilon}_{m,k}^H\mathbf{v}_m+2\zeta_{m,k}\Re\left\{\bm{\upsilon}_{m,k}^H\mathbf{v}_m\right\}+\zeta_{m,k}^2, \\
    &\Big\Vert\sum_{m\in\mathcal{A}}{\mathbf{Q}_k{\widetilde{\mathbf{H}}}_m\mathbf{V}_m{\widetilde{\mathbf{H}}}_m^H}+\mathbf{Q}_k\mathbf{D}_k-\mathbf{I}+\rho\mathbf{Z}_{3,k} \Big\Vert^2\nonumber\\
    =&\Big\Vert\underbrace{\left({\widetilde{\mathbf{H}}}_m^\ast\otimes\mathbf{Q}_k{\widetilde{\mathbf{H}}}_m\right)}_{\triangleq \mathbf{\Gamma}_{m,k}}\mathbf{v}_m\nonumber\\
    &+\underbrace{\mathrm{vec}\Big(\sum_{m^\prime\neq m}{\mathbf{Q}_k{\widetilde{\mathbf{H}}}_{m^\prime}\mathbf{V}_{m^\prime}{\widetilde{\mathbf{H}}}_{m^\prime}^H}\!+\!\mathbf{Q}_k\mathbf{D}_k\!-\!\mathbf{I}\!+\!\rho\mathbf{Z}_{3,k}\Big)}_{\triangleq \bm{\theta}_{m,k}}\Big\Vert^2\nonumber\\
    =&\mathbf{v}_m^H\mathbf{\Gamma}_{m,k}^H\mathbf{\Gamma}_{m,k}\mathbf{v}_m\!+\!2\Re\!\left\{\bm{\theta}_{m,k}^H\mathbf{\Gamma}_{m,k}\mathbf{v}_m\right\}\!+\!\Vert\bm{\theta}_{m,k}\Vert^2,\\
   &\Big\Vert\mathbf{V}_m\!-\!{\widetilde{\mathbf{V}}}_m\!+\!\rho\mathbf{Z}_{4,m}\Big\Vert^2\!=\!\Big\Vert\mathbf{v}_m-\underbrace{\mathrm{vec}\left({\widetilde{\mathbf{V}}}_m-\rho\mathbf{Z}_{4,m}\right)}_{\triangleq \bm{\psi}_m}\Big\Vert^2\nonumber\\
   =&\mathbf{v}_m^H\mathbf{v}_m\!-\!2\Re\!\left\{\bm{\psi}_m^H\mathbf{v}_m\right\}\!+\!\Vert\bm{\psi}_m\Vert^2,
\end{align}
where $\mathbf{v}_m \triangleq \mathrm{vec}(\mathbf{V}_m)$ and $\bm{\upsilon}_{m,k} \triangleq \mathrm{vec}(\mathbf{\Upsilon}_{m,k})$. Based on the above transformations,  the optimal solution is derived as
\begin{align}\label{eq48}
\mathbf{V}_m^\ast\!=\!\mathrm{mat}&\left(\!\left(\sum_{k=1}^{K}\left(\bm{\upsilon}_{m,k}\bm{\upsilon}_{m,k}^H+\mathbf{\Gamma}_{m,k}^H\mathbf{\Gamma}_{m,k}\right)+\mathbf{I}\right)^{\dag}\right.\nonumber\\
&\times\!\left.\left(\!\bm{\psi}_m\!-\!\sum_{k=1}^{K}\left(\zeta_{m,k}\bm{\upsilon}_{m,k}\!+\!\mathbf{\Gamma}_{m,k}^H\bm{\theta}_{m,k}\right)\!\right)\!\right).
\end{align}

\subsubsection{Updating ${\widetilde{\mathbf{V}}}_m$}  The subproblem regarding ${\widetilde{\mathbf{V}}}_m$ is expressed as 
\begin{align}\label{BCD3}
   \mathop{\text{minimize}}_{{\widetilde{\mathbf{V}}}_m}\quad&\big\Vert\mathbf{V}_m-{\widetilde{\mathbf{V}}}_m+\rho\mathbf{Z}_{4,m}\big\Vert_F^2\nonumber\\
   \text{subject to}\quad & \mathrm{Tr}\left({\widetilde{\mathbf{V}}}_m\right)\le P_m,  \nonumber\\
   & {\widetilde{\mathbf{V}}}_m\succcurlyeq\mathbf{0}, 
\end{align}
which is also a convex problem. We derive its optimal solution in the following proposition.

\begin{proposition} \label{proposition2}
  The optimal solution to ${\widetilde{\mathbf{V}}}_m$ is equal to
\begin{align}\label{eq50}
{\widetilde{\mathbf{V}}}_m^\ast=\mathbf{U}_m\mathrm{diag}\left(x_{m,1},\cdots,x_{m,N_a}\right)\mathbf{U}_m^H,
\end{align}
where $\mathbf{U}_m$ is the singular matrix of $\frac{\mathbf{V}_m+\rho\mathbf{Z}_{4,m}+\mathbf{V}_m^H+\rho\mathbf{Z}_{4,m}^H}{2}$. Its optimal eigenvalue is equal to
\begin{align}\label{opteig}
    x_{m,n}^\ast=
    \begin{cases}
    {\bar{y}}_{m,n},&S_m\le P_m,\\
    \max\{{y_{m,n}-\tau,0}\},&S_m>P_m,
    \end{cases}
\end{align}
where $y_{m,n}$ is the eigenvalue of $\frac{\mathbf{V}_m+\rho\mathbf{Z}_{4,m}+\mathbf{V}_m^H+\rho\mathbf{Z}_{4,m}^H}{2}$, ${\bar{y}}_{m,n}\triangleq \max\{{y_{m,n},0}\}$, $S_m\triangleq \sum_{n=1}^{N_a}{\bar{y}}_{m,n}$, and $\tau$ is a positive number satisfying $\sum_{n=1}^{N_a}x_{m,n}^\ast=P_m$.
\end{proposition}

\begin{proof}
    Please refer to Appendix \ref{app3}. \hfill $\square$
\end{proof}

 \subsubsection{Updating $\{\tau_k,u_k\}_{k=1}^K$}
 We simplify the optimization problem with respect to $\{\tau_k,u_k\}_{k=1}^K$ as 
\begin{align}
   \mathop{\text{maximize}}\limits_{\{\tau_k,u_k\}}\quad\min_{k}&\quad{\tau_k-u_k-\frac{1}{2\rho}\times}\nonumber\\
   &\sum_{k=1}^{K}\!\!\bigg(\!\!\left(2^{\tau_k}\!\!-\!\!\alpha_k\right)^2\!\!+\!\left(u_k\!-\!\chi_k\!+\!\rho z_{5,k}\right)^2\!\bigg),
\end{align}
where $\alpha_k \triangleq \left(1\!+\!p_k{\widetilde{\mathbf{h}}}_k^H\mathbf{Q}_k{\widetilde{\mathbf{h}}}_k\!+\!\rho z_{1,k}\right)$. It is equivalent to
\begin{align}
   \mathop{\text{maximize}}_{\left\{\tau_k,u_k\right\},r}\quad&r\!-\!\frac{1}{2\rho}\sum_{k=1}^{K}\left(\left(2^{\tau_k}\!-\!\alpha_k\right)^2\!+\!\left(u_k\!-\!\chi_k\!+\!\rho z_{5,k}\right)^2\right)\nonumber\\
   \text{subject to}\quad & r\le\tau_k-u_k, \forall k
\end{align}
and it is a convex problem. By checking the KKT conditions, the optimal solutions admit the following forms:
\begin{align}\label{eq54}
    &\tau_k^\ast =\log_2{\frac{\alpha_k+\sqrt{\alpha_k^2+\frac{4\rho\lambda_k}{\ln{2}}}}{2}},\nonumber\\
    &u_k^\ast=\chi_k-\rho z_{5,k}-\rho\lambda_k,\nonumber\\
    &\sum_{k=1}^{K}\lambda_k=1,\lambda_k\geq0,\nonumber\\
    &\lambda_k\left(r^\ast-\tau_k+u_k\right)=0,
\end{align}
where $\lambda_k$ is the dual variable. 
For the optimal solutions $\left\{\tau_k^\ast, u_k^\ast\right\}$ and $r^\ast$, if $\tau_k^\ast - u_k^\ast > r^\ast$, then $\tau_k^\ast = \log_2{\alpha_k}$, $u_k^\ast = \chi_k - \rho z_{5,k}$. Based on this finding, we assume without loss of generality that $\log_2{\alpha_1} - \chi_1 + \rho z_{5,1}\geq, \cdots,\geq  \log_2{\alpha_K} - \chi_K + \rho z_{5,K}$. Suppose there exists an index $\bar{k}$ such that for any $k \le \bar{k}$, $\tau_k^\ast = \log_2\alpha_k$ and $u_k^\ast = \chi_k - \rho z_{5,k}$; while for any $k > \bar{k}$, $\tau_k^\ast = r^\ast + u_k^\ast$. Moreover, when $k > \bar{k}$, for each given $r$, the corresponding $\lambda_k$ can be uniquely determined. Using the bisection search procedure to ensure $\sum_{k=\bar{k}+1}^{K}\lambda_k=1$, we can obtain the corresponding $r^\ast$ and all other associated variables. By enumerating all possible $\bar{k}$, the solution yielding the largest objective value is selected as the final optimal solution.

 \subsubsection{Updating $\chi_k$} Finally, we consider the optimization over $\chi_k$, which follows
\begin{align}\label{eq55}
    \mathop{\text{minimize}}_{\chi_k}&\ \left(\frac{1}{2^{\chi_k}\!-\!1}\!-\!{\beta_k}\right)^2+\left(u_k-\chi_k+\rho z_{5,k}\right)^2,
\end{align}
and $\beta_k\triangleq c_k\!+\!\frac{\sum_{m\in\mathcal{A}}\!\mathbf{w}_{e,k}^H\mathbf{G}_m\mathbf{V}_m\mathbf{G}_m^H\mathbf{w}_{e,k}}{{\widetilde{p}}_k}\!\!-\!\frac{1}{2^{\chi_k}-1}\!+\!\rho z_{2,k} $. However, for such a single-variable optimization problem, it can be easily verified by taking the derivative that the problem admits at least one minimum point. Therefore, it can be efficiently solved using the Newton method.

After the BCD iterations converge, the dual variables and penalty parameters are updated in the same manner as described in Sec. \ref{SEC3}, and thus the details are omitted here for brevity. In summary, the complete procedure for solving problem (\ref{eq36}) is presented in Algorithm \ref{alg2}. In the BCD iterations, the overall computational cost is mainly dominated by the update of $\mathbf{Q}_{k}$, which involves a matrix inversion with the largest dimension. This step incurs a complexity on the order of $\mathcal{O}(K N ^6 N_a^6)$.

\begin{algorithm}[!t]
\caption{Proposed algorithm for solving (\ref{eq36})} \label{alg2}
\begin{algorithmic}[1]  
\STATE \textbf{Initialize} feasible initial variables, dual variables,  penalty parameter $\rho^{(0)}$, constants $c \in(0,1)$ and $\ell=0$.
\REPEAT
\REPEAT
\STATE Update $\mathbf{Q}_k$, $\forall k$ according to (\ref{eq43}).
\FOR{$m=1,\cdots,M$}
\STATE Update $\mathbf{V}_m$ according to (\ref{eq48}).
\ENDFOR
\STATE Update $\widetilde{\mathbf{V}}_m$, $\forall m$ according to (\ref{eq50}).
\STATE Update $\{\tau_k,u_k\}_{k=1}^K$ according to (\ref{eq54}).
\STATE Update $\chi_k$, $\forall k$ by solving problem (\ref{eq55}).
\UNTIL convergence.
\STATE Update the dual variables.
\STATE Update the penalty parameter as $\rho^{(\ell+1)}=c\rho^{(\ell)}$.
\STATE Set $\ell\gets\ell+1$.
\UNTIL convergence.
\end{algorithmic} 
\end{algorithm}

\begin{figure}[!t]
\centering
\includegraphics[width=3.2in]{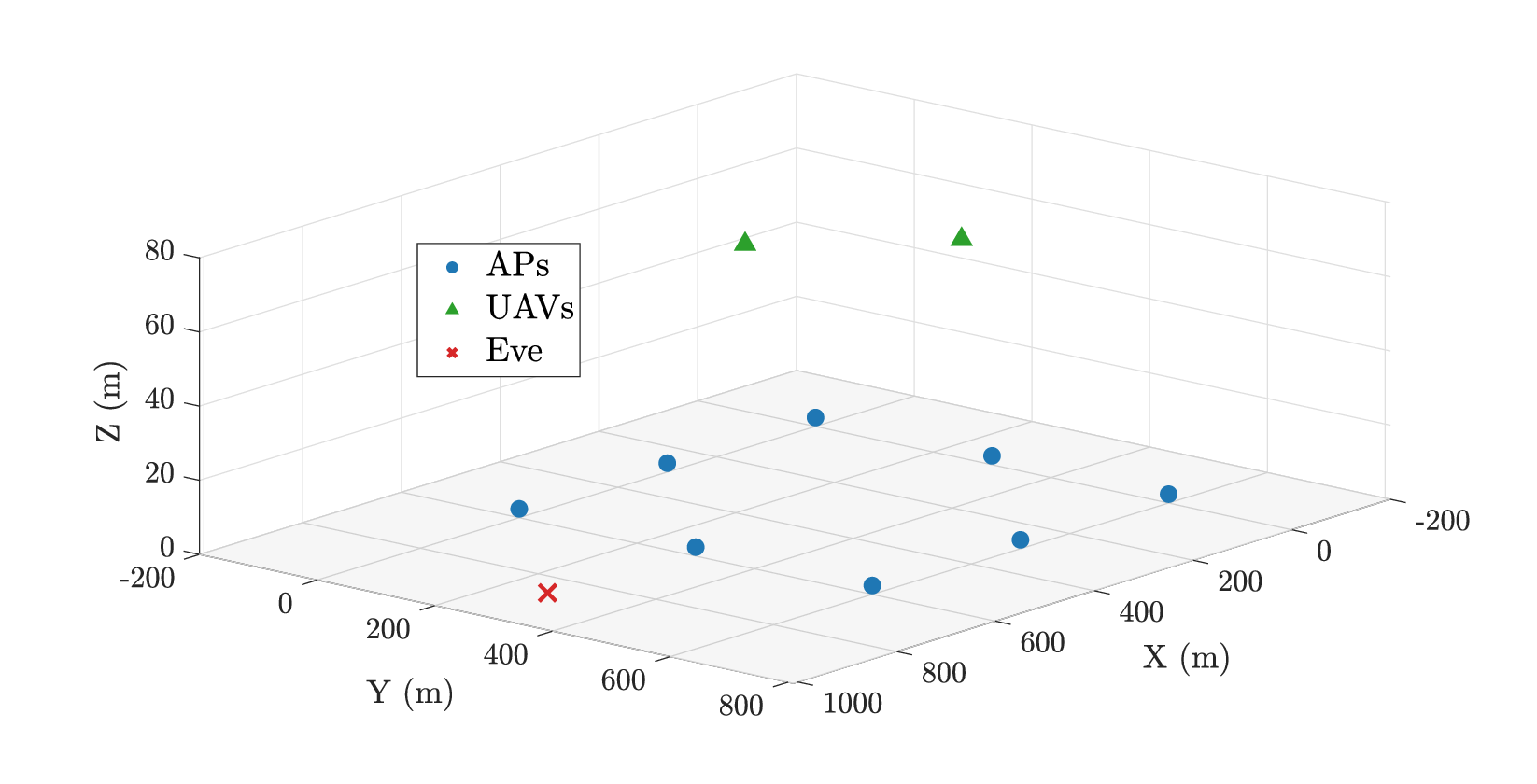}
\caption{Three-dimensional topology of the considered LAWN.}\label{layout}
\end{figure}

\section{Numerical Results}
As illustrated in Fig.~\ref{layout}, we consider a CF-based LAWN in which multiple distributed APs cooperatively serve uplink UAVs in the presence of a multi-antenna eavesdropper. Unless otherwise specified, the simulation parameters are as follows: the network comprises $M=8$ APs, each equipped with $N_a = 4$ antennas and deployed at fixed ground coordinates $(0,0,0)$, $(300,0,0)$, $(600,0,0)$, $(0,300,0)$, $(600,300,0)$, $(0,600,0)$, $(300,600,0)$, and $(600,600,0)$. Two single-antenna UAVs, each transmitting with $p_k = 30~\text{dBm}$, are randomly positioned within a 3D semi-ellipsoidal region centered at $(300,300,0)$ with a horizontal radius of $600~\text{m}$ and a vertical height of $60~\text{m}$. A passive eavesdropper equipped with $N_e = 4$ antennas is located at $(900,300,0)$ and employs an MMSE receiver for signal detection. Each T-AP is constrained by a maximum AN transmit power of $P_m = 30~\text{dBm}$, while the noise powers at both APs and Eve are set to $\sigma_n^2 = \sigma_e^2 = -57~\text{dBm}$. The parameter $\varpi$ is chosen as $\varpi = 0.05$.

\begin{figure}[!t]
\centering
\includegraphics[width=3.2in]{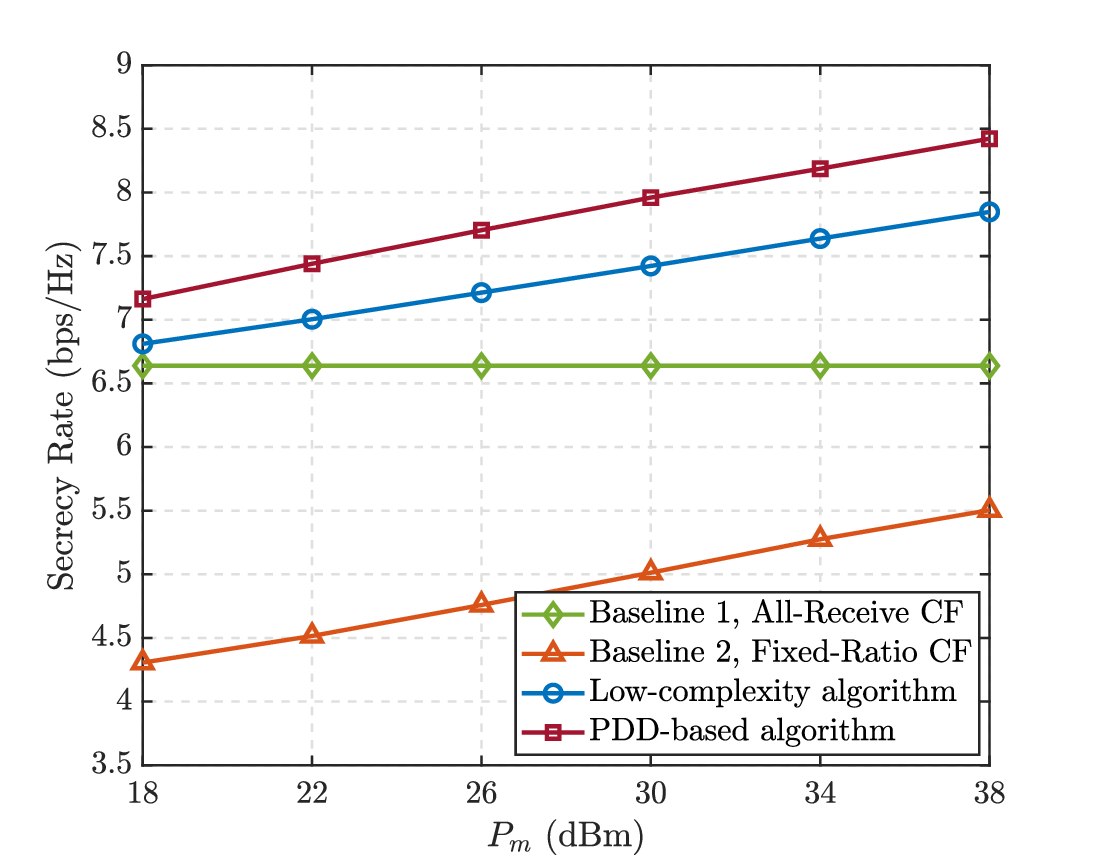}
\caption{Secrecy rate versus AN power budget, $P_m$.}\label{fig3}
\end{figure}

\begin{figure}[!t]
\centering
\includegraphics[width=3.2in]{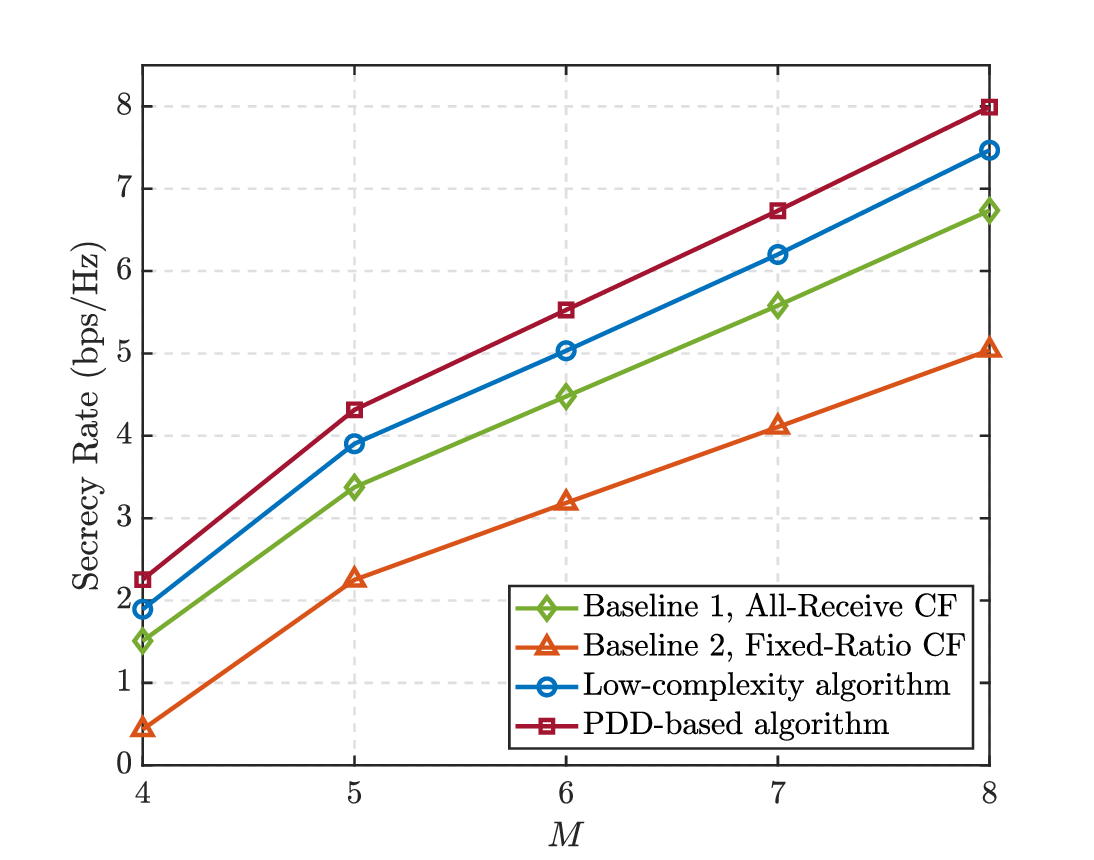}
\caption{Secrecy rate versus the number of APs, $M$.}\label{fig4}
\end{figure}

All air-to-ground and ground-to-ground channels follow a standard mmWave model \cite{lin2017subarray}, where LoS propagation dominates due to the elevated UAV altitude and the sparse scattering characteristics of the environment, while NLoS components are generated according to the sparse multipath structure typical of mmWave channels. All results are averaged over a large number of independent channel realizations.

To benchmark the proposed flexible-duplex scheme, we consider the following two baselines:

\begin{itemize}
\item \textbf{Baseline~1 (All-Receive CF):} All APs operate as R-APs, i.e., $x_m^C = 1$ and $x_m^J = 0$, $\forall m$.
\item \textbf{Baseline~2 (Fixed-Ratio CF):} Half of the APs are set as R-APs, while the remaining half operate as T-APs.
\end{itemize}

These baselines allow us to evaluate the performance gains brought by AP-level duplex flexibility and by the joint optimization of mode selection and AN design.

Fig.~\ref{fig3} shows the secrecy rate as the maximum AN power budget $P_m$ increases from $18$ to $38$~dBm. For both proposed methods, the secrecy rate improves steadily with larger $P_m$, since a higher AN budget enables stronger jamming toward Eve without excessively degrading the legitimate uplink after optimization. In contrast, Baseline~1 remains constant because no AN is transmitted, and thus increasing $P_m$ has no effect. Baseline~2 exhibits inferior performance, demonstrating that arbitrarily assigning a fixed portion of APs can lead to severe AN leakage and even harm secrecy performance when mode selection is not jointly optimized. Across the entire range, the proposed joint optimization achieves the highest secrecy rate, while the low-complexity sequential scheme closely follows and maintains consistently large gains over both baselines.

\begin{figure}[!t]
\centering
\includegraphics[width=3.2in]{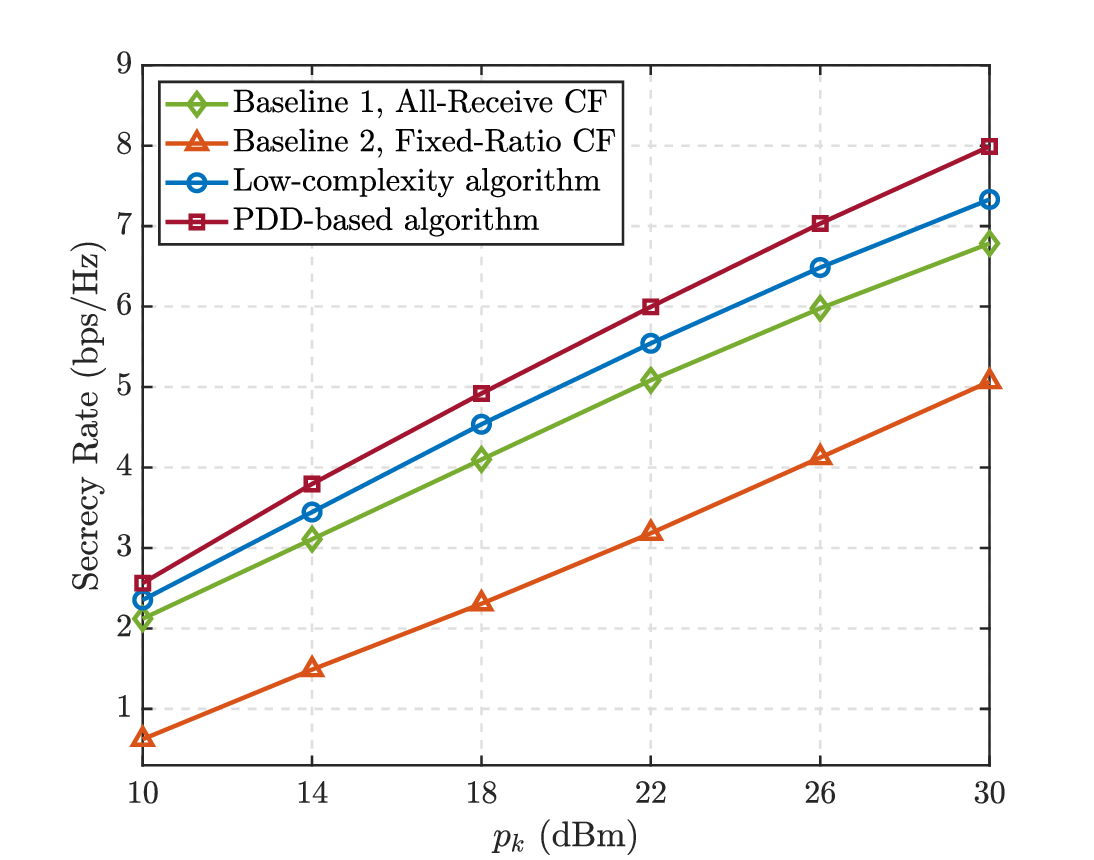}
\caption{Secrecy rate versus UAV transmit power, $p_k$.}\label{fig5}
\end{figure}

\begin{figure}[!t]
\centering
\includegraphics[width=3.2in]{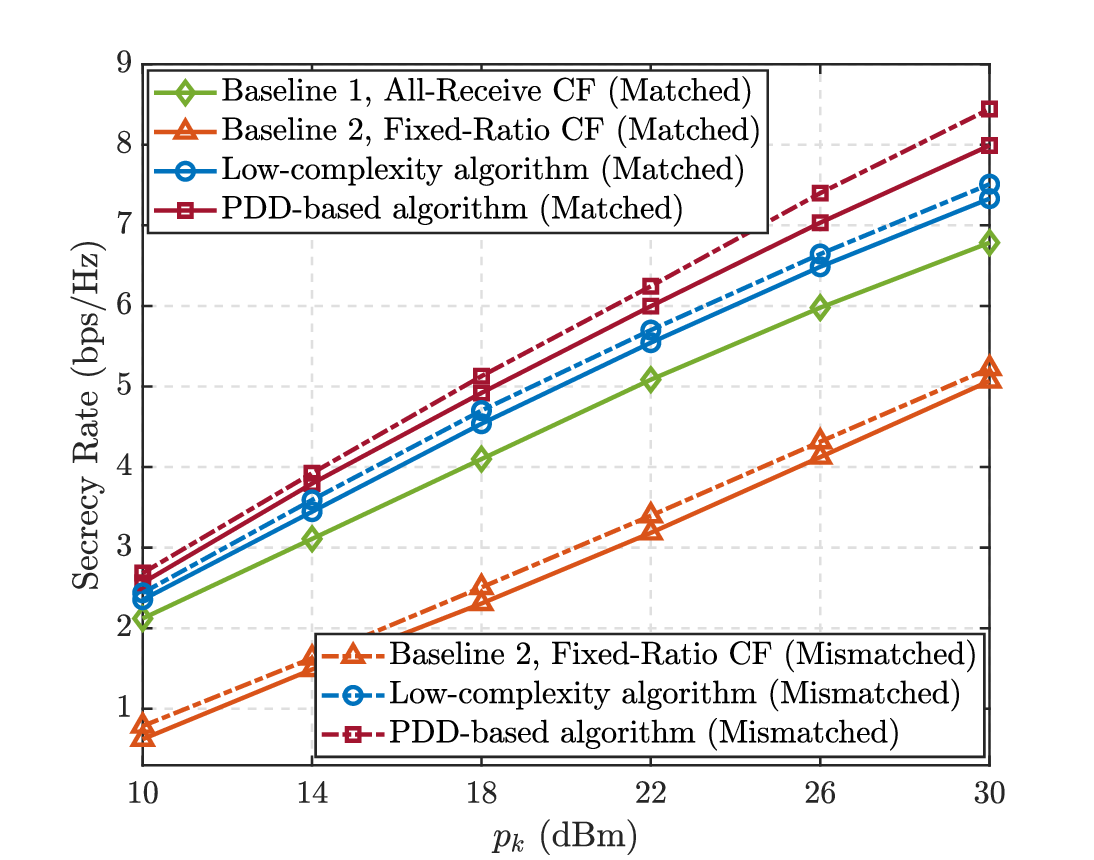}
\caption{Secrecy rate under matched/mismatched assumptions.}\label{fig6}
\end{figure}

Fig.~\ref{fig4} plots the secrecy rate as a function of the number of APs, $M$. It is observed that the secrecy rate increases consistently with more APs for both proposed schemes and the baselines. This improvement can be attributed to two main factors: (i) additional APs provide higher spatial diversity gains in cooperative reception, enhancing the legitimate signal quality; and (ii) additional APs enable more APs to be flexibly assigned as T-APs for AN transmission, effectively suppressing the eavesdropper. Notably, when $M=4$, the gain of our schemes over Baseline~2 is less than $2~\text{bps/Hz}$, whereas with $M=8$, the gain increases to nearly $3~\text{bps/Hz}$, reflecting that the joint optimization of mode selection and AN design becomes increasingly beneficial as more APs are available.

Fig.~\ref{fig5} presents the secrecy rate as the UAV transmit power $p_k$ increases from $10$ to $30~\text{dBm}$. As expected, all schemes benefit from higher UAV power, since stronger desired signals improve the uplink SINR at the R-APs. This allows the system to allocate more APs as T-APs for AN transmission, thereby boosting secrecy performance. Moreover, the performance gap between the proposed schemes and the baselines widens with increasing $p_k$, as the joint optimization more effectively exploits the stronger legitimate signals to suppress the eavesdropper. Among all schemes, the PDD-based joint optimization achieves the highest secrecy rate, while the low-complexity heuristic maintains over $90\%$ of this performance.

In Figs.~\ref{fig3}-\ref{fig5}, the AN design is optimized under the assumption that Eve employs an MMSE receiver, i.e., the optimal linear detector. Fig.~\ref{fig6} further evaluates the secrecy rate under matched/mismatched eavesdropper receiver assumptions. The curve labeled ``Matched" corresponds to the scenario where Eve indeed uses MMSE, whereas the ``Mismatched" curve represents the case where Eve instead adopts a ZF receiver. The mismatched configuration achieves a higher secrecy rate for two reasons. First, ZF generally underperforms MMSE due to its susceptibility to noise amplification, which weakens Eve’s detection capability. Second, the proposed AN design exhibits strong robustness, as the optimized AN continues to effectively suppress Eve’s signal reception even when the actual receiver differs from the assumed model. These results confirm the robustness of the proposed flexible-duplex CF framework and highlight the importance of assuming the optimal MMSE receiver during optimization to guarantee worst-case secrecy performance under potential model mismatches.

\begin{figure}[!t]
\centering
\includegraphics[width=3.2in]{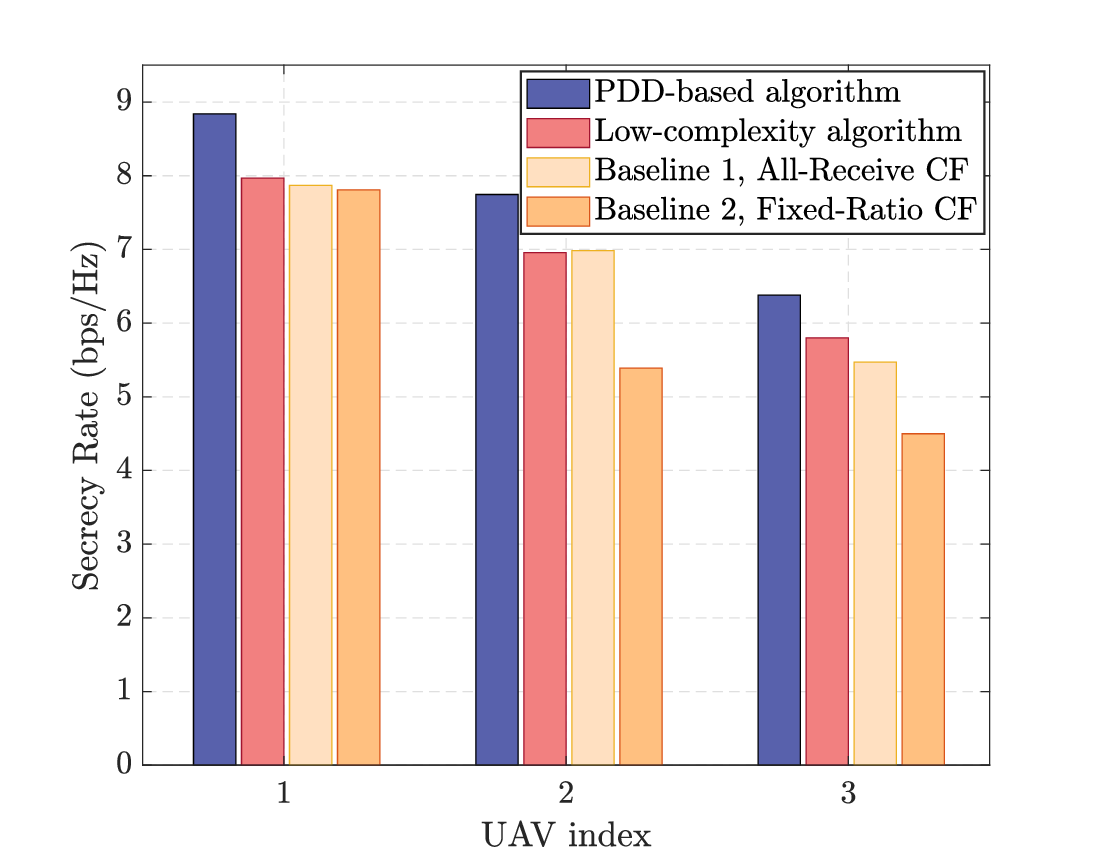}
\caption{Fairness evaluation of UAV secrecy rates.}\label{fig7}
\end{figure}
Fig.~\ref{fig7} illustrates the secrecy rates of individual UAVs under the proposed PDD-based scheme, the low-complexity scheme, and the two baselines. As observed, the PDD algorithm, which maximizes the minimum secrecy rate among all UAVs, achieves substantial gains across the board. In particular, compared with two baselines, the proposed PDD-based method significantly improves the secrecy rates of all UAVs. Furthermore, by adopting the max-min secrecy rate criterion, the low-complexity scheme maintains relatively stable performance across the UAVs. Consequently, the gap between the best and worst UAVs is narrower, demonstrating a more balanced and fair distribution of secrecy rates.

\section{Conclusion}
This work proposed a flexible-duplex CF architecture for secure uplink communications in LAWNs, where each AP can operate in either reception or jamming mode. A joint optimization strategy with closed-form combiners and a PDD-based algorithm was developed, alongside a sequential low-complexity scheme to improve scalability. Simulation results demonstrate that the flexible-duplex design significantly improves secrecy rates over baselines, while the low-complexity scheme offers an effective performance–complexity tradeoff suitable for large-scale deployments.

\begin{appendices}

\section{Proof of Proposition \ref{proposition1}} \label{appb}
We first note that if the $m$-th AP is operated as a T-AP (i.e., $x_m^C =0$), it no longer performs signal reception, and thus its corresponding combining vector satisfies $\mathbf{u}_{m,k}=\mathbf{0}$. Hence, we can conclude that
\begin{align}
    \mathbf{u}_k^H\mathbf{S}_c\mathbf{u}_k \!=\!\sum_{m=1}^M x_m^C \Vert \mathbf{u}_{m,k}\Vert ^2 \!=\!\sum_{m=1}^M  \Vert \mathbf{u}_{m,k}\Vert ^2 = \Vert \mathbf{u}_{k}\Vert ^2. 
\end{align}
By substituting $\mathbf{u}_k^H\mathbf{S}_c\mathbf{u}_k$ with $\Vert \mathbf{u}_{k}\Vert ^2$, the SINR term in (\ref{rsinr}) is reformulated as  
\begin{align}
    \gamma_k= \frac{p_k\left|\mathbf{u}_k^H\mathbf{S}_c\mathbf{h}_k\right|^2}{\mathbf{u}_k^H\mathbf{\Sigma}_k\mathbf{u}_k},
\end{align}
where $\mathbf{\Sigma}_k\!\triangleq \!\sum_{i\neq k}^{K}{p_i\mathbf{S}_c\mathbf{h}_i\mathbf{h}_i^H\mathbf{S}_c}\!+\!\sum_{m=1}^{M}\!{x_m^J\mathbf{S}_c\mathbf{H}_m\!\mathbf{V}_m\mathbf{H}_m^H\mathbf{S}_c}\!+\!\sigma_n^2\mathbf{I}$. Now, the maximization of $\gamma_k$  belongs to the problem of generalized Rayleigh quotient, and the optimal $\mathbf{u}_k$ is equal to $p_k \mathbf{\Sigma}_k^{-1}\mathbf{S}_c \mathbf{h}_k$ \cite{yaotwc,rayq}. The proof completes.

\section{Proof of Proposition \ref{proposition2}} \label{app3}
We first decompose $\mathbf{V}_m+\rho \mathbf{Z}_{4,m}$ into the sum of Hermitian and skew-Hermitian matrices, i.e., $ \mathbf{A}_m\!+\!\mathbf{B}_m$, where 
\begin{align}
    &\mathbf{A}_m \triangleq \frac{\mathbf{V}_m+\rho\mathbf{Z}_{4,m}+\mathbf{V}_m^H+\rho\mathbf{Z}_{4,m}^H}{2},\nonumber \\
    & \mathbf{B}_m \triangleq \frac{\mathbf{V}_m+\rho\mathbf{Z}_{4,m}-\mathbf{V}_m^H-\rho\mathbf{Z}_{4,m}^H}{2}.
\end{align}
Then, the objective in (\ref{BCD3}) is rewritten as
\begin{align}
    &\big\Vert\mathbf{V}_m-{\widetilde{\mathbf{V}}}_m+\rho\mathbf{Z}_{4,m}\big\Vert_F^2 =\big\Vert{\widetilde{\mathbf{V}}}_m-\mathbf{A}_m-\mathbf{B}_m\big\Vert_F^2  \nonumber \\
    =&\big\Vert\!{\widetilde{\mathbf{V}}}_m\!\!-\!\!\mathbf{A}_m\!\big\Vert_F^2\!+\!\!\big\Vert\mathbf{B}_m\!\big\Vert_F^2\!-\!2\Re\!\!\left\{\!\mathrm{Tr}\!\left(\! ({\widetilde{\mathbf{V}}}_m\!\!-\!\!\mathbf{A}_m)^H\mathbf{B}_m\right)\! \right\}.
\end{align}
We note that ${\widetilde{\mathbf{V}}}_m-\mathbf{A}_m$ is a Hermitian matrix while $\mathbf{B}_m$ is a skew-Hermitian matrix. Hence, it is easy to verify that $\Re\left\{\mathrm{Tr}\left( ({\widetilde{\mathbf{V}}}_m-\mathbf{A}_m)^H\mathbf{B}_m\right) \right\}=0$. Equivalently, the subproblem in (\ref{BCD3}) is reformulated as
\begin{align}
    \mathop{\text{minimize}}_{{\widetilde{\mathbf{V}}}_m\succeq \mathbf{0}}\quad&\big\Vert {\widetilde{\mathbf{V}}}_m-\mathbf{A}_{m}\big\Vert_F^2\nonumber\\
   \text{subject to}\quad & \mathrm{Tr}\left({\widetilde{\mathbf{V}}}_m\right)\le P_m.
\end{align}
Based on the unitary invariance of the Frobenius norm, we equivalently transform the above problem as
\begin{align}
    \mathop{\text{minimize}}_{x_{m,1}, \cdots,x_{m,N_a}}\quad&\sum_{n=1}^{N_a}\left(x_{m,n}-y_{m,n}\right)^2\nonumber\\
   \text{subject to}\quad & \sum_{n=1}^{N_a}x_{m,n}\le P_m, \forall m, \nonumber\\
   & x_{m,n}\geq0, \forall m.
\end{align}
The optimal solution to this problem is quite straightforward, as illustrated in (\ref{opteig}), and we complete the proof.

\end{appendices}

\vspace{-5.pt}

\bibliographystyle{IEEEtran}
\bibliography{IEEEabrv,reference}




\end{document}